\documentclass[11pt]{article}
\usepackage{amscd}
\usepackage{graphics}
\usepackage{color}
\usepackage{cancel}
\usepackage{stmaryrd}
\usepackage{tikz}
\usetikzlibrary{matrix,arrows}
\usepackage{amsmath}
\usepackage{amsthm}
\usepackage{amssymb}
\usepackage{proof}
\usepackage{verbatim,cmll}
\usepackage{setspace}
\usepackage{lscape}
\usepackage{latexsym,relsize}
\usepackage{amscd}
\usepackage{multicol}
\usepackage{bussproofs}
\usepackage{txfonts}
\usepackage{tikz}
\usepackage[position=top]{subfig}
\usepackage{leqno}
\usepackage{authblk}
\usepackage{marvosym}
\DeclareSymbolFont{extraup}{U}{zavm}{m}{n}
\DeclareMathSymbol{\varheart}{\mathalpha}{extraup}{86}
\DeclareMathSymbol{\vardiamond}{\mathalpha}{extraup}{87}
\DeclareMathSymbol{\vardiamond}{\mathalpha}{extraup}{87}

\newcommand{\commment}[1]{}

\renewcommand{\phi}{\varphi}
\renewcommand{\emptyset}{\varnothing}

\renewcommand{\epsilon}{\varepsilon}

\newcommand{\nomi}{\mathbf{i}}
\newcommand{\nomj}{\mathbf{j}}
\newcommand{\nomk}{\mathbf{k}}

\theoremstyle{plain}
\newtheorem{thm}{Theorem}
\newtheorem{theorem}{Theorem}[section]

\newtheorem{corollary}[theorem]{Corollary}
\newtheorem{example}[theorem]{Example}

\newtheorem{proposition}[thm]{Proposition}

\theoremstyle{definition}

\newtheorem{definition}[thm]{Definition}

\title{Sahlqvist Correspondence Theory for Second-Order Propositional Modal Logic}
\author{Zhiguang Zhao}

\date{}

\begin{document}
\maketitle
\begin{abstract}Modal logic with propositional quantifiers (i.e.\ second-order propositional modal logic ($\mathsf{SOPML}$)) has been considered since the early time of modal logic. Its expressive power and complexity are high, and its van-Benthem-Rosen theorem and Goldblatt-Thomason theorem have been proved by ten Cate (2006). However, the Sahlqvist theory of $\mathsf{SOPML}$ has not been considered in the literature. In the present paper, we fill in this gap. We develop the Sahlqvist correspondence theory for $\mathsf{SOPML}$, which covers and properly extends existing Sahlqvist formulas in basic modal logic. We define the class of Sahlqvist formulas for $\mathsf{SOMPL}$ step by step in a hierarchical way, each formula of which is shown to have a first-order correspondent over Kripke frames effectively computable by an algorithm $\mathsf{ALBA}^{\mathsf{SOMPL}}$. In addition, we show that certain $\Pi_2$-rules correspond to $\Pi_2$-Sahlqvist formulas in $\mathsf{SOMPL}$, which further correspond to first-order conditions, and that even for very simple $\mathsf{SOMPL}$ Sahlqvist formulas, they could already be non-canonical.\\

\end{abstract}

\section{Introduction}

\paragraph{Second-Order Propositional Modal Logic ($\mathsf{SOMPL}$).}Modal logic with propositional quantifiers has been considered in the literature since Kripke \cite{Kr59}, Bull \cite{Bu69}, Fine \cite{Fi69,Fi70}, and Kaplan \cite{Ka70}.\footnote{For more literature, see \cite{AnTh02,BevdH15,BevdH16,BevDvdH16,BeGuMa97,Di18,Fr20,GhZa95,Ho19,HoLi18,Kr97,Ku04,Ku08,Ku15}.} This language is of high complexity: its satisfiability problem is not decidable, and indeed not even analytical. In Kaminski and Tiomkin \cite{KaTi96}, the authors showed that the expressive power for $\mathsf{SOMPL}$ whose modalities are S4.2 or weaker is the same as second-order predicate logic. However, not every second-order formula is equivalent to an $\mathsf{SOMPL}$-formula, \mbox{since} $\mathsf{SOMPL}$-formulas are preserved under generated submodels (see van Benthem \cite{vB83}). In ten Cate \cite{tC06}, the author proved the analogues of the van Benthem-Rosen theorem (on the model level) and Goldblatt-Thomason theorem (on the frame level) for $\mathsf{SOMPL}$. Therefore, a natural question is: on the frame level, can we find a natural fragment of $\mathsf{SOPML}$-formulas such that each formula in this fragment corresponds to a first-order formula, in the sense of Sahlqvist theory (see \cite{Sa75,vB83})? This is what we will answer in the paper.

\paragraph{Correspondence Theory.}Typically, modal correspondence theory \cite{vB83} \mbox{concerns} the correspondence of modal formulas and first-order formulas over Kripke frames, via the tools of standard translation. Syntactic classes (e.g.\ Sahlqvist formulas \cite{Sa75}, inductive formulas \cite{GorankoV06}, etc.) of modal formulas are identified to have first-order correspondents and are canonical, i.e.\ their validity are closed under taking canonical extensions. 

In the present paper, we identify the Sahlqvist formulas of $\mathsf{SOMPL}$, which cover and properly extend the Sahlqvist fragment in basic modal logic. We show that each Sahlqvist $\mathsf{SOMPL}$ formula corresponds to a first-order formula by an algorithm $\mathsf{ALBA}^{\mathsf{SOPML}}$. In particular, we have the following observations: the $\mathsf{SOMPL}$ Sahlqvist formula $\forall p(\Box p\land\forall q(q\to \Diamond\Diamond q\lor p)\to p)$ corresponds to $\forall x\forall y(Rxy\land Ryx\to Rxx)$, which is not modally definable since this property is not preserved under taking bounded morphic image (see Example \ref{Exa:Non-Sahlqvist}); the $\mathsf{SOMPL}$ Sahlqvist formula $\forall q(\forall p(p\to\Diamond p\lor q)\to q)$ is not canonical (see Example \ref{Example:irreflexivity}), which is in contrast to the basic modal logic setting where each Sahlqvist formula is canonical.

\paragraph{Non-standard Rules.}Another topic that is related to the present paper is non-standard rules, starting from Gabbay \cite{Ga81} where a non-standard rule for irreflexivity is introduced. These rules have been used in temporal logic \cite{Bu80,GaHo90}, region-based theories of space \cite{BaTiVa07,Va07} and are used to prove completeness results for modal logic systems with non-$\xi$-rules \cite{Ve93}. In particular, the so-called $\Pi_2$-rules \cite{BeBeSaVe19,BeGhLa20,Sa16} which generalize both the irreflexivity rule of Gabbay \cite{Ga81} and the non-$\xi$-rules of Venema \cite{Ve93}, have their natural $\forall\exists$-counterparts, which are essentially $\forall\exists$-$\mathsf{SOMPL}$ formulas, fit naturally into the language of $\mathsf{SOMPL}$. We use the correspondence algorithm to compute the first-order correspondents of a subclass of $\Pi_2$-rules whose $\forall\exists$-counterparts are $\mathsf{SOMPL}$ $\Pi_2$-Sahlqvist formulas.

\paragraph{Our methodology.}The present paper use the same methodology as \cite{ConPalSou,CoGhPa14}. In the present paper, inspired by the Sahlqvist rules in Santoli \cite{Sa16}, we identify the Sahlqvist formulas of $\mathsf{SOMPL}$, which are generalizations of Sahlqvist formulas in basic modal logic and have first-order correspondents. The Sahlqvist fragment of $\mathsf{SOPML}$ is defined in a step-by-step way, and we give an algorithm $\mathsf{ALBA}^{\mathsf{SOPML}}$ (Ackermann Lemma Based Algorithm) which can successfully reduce Sahlqvist formulas in $\mathsf{SOPML}$ to first-order formulas and is sound with respect to Kripke semantics.

\paragraph{Structure of the paper.}The structure of the paper is as follows: Section \ref{Sec:Prelim} gives the necessary preliminaries. Section \ref{Sec:Sahl} gives the definition of Sahlqvist $\mathsf{SOPML}$ formulas step by step. Section \ref{Sec:ALBA} defines the algorithm $\mathsf{ALBA}^{\mathsf{SOPML}}$. Section \ref{Sec:Soundness} shows the soundness of the algorithm with respect to Kripke frames. Section \ref{Sec:Success} shows that the algorithm succeeds on all Sahlqvist $\mathsf{SOPML}$ formulas. Section \ref{Sec:Example} gives some examples and connect them with non-standard rules, and one example shows that even for very simple Sahlqvist $\mathsf{SOPML}$ formulas, they can already be non-canonical. Section \ref{Sec:Conclusion} gives some final remarks and conclusion.

\section{Preliminaries}\label{Sec:Prelim}

\subsection{Language and semantics}

In the present paper we consider the unimodal language. Given a set $\mathsf{Prop}$ of propositional variables, the second-order propositional modal formulas are defined as follows:
$$\phi::=p\mid \bot\mid \top\mid \neg\phi\mid \phi\land\phi\mid \phi\lor\phi\mid \phi\to\phi\mid\Box\phi\mid\Diamond\phi\mid \forall p\phi\mid \exists p\phi$$

where $p\in\mathsf{Prop}$. We use the notation $\vec p$ to denote a set of propositional variables and use $\phi(\vec p)$ to indicate that the propositional variables occur in $\phi$ are all in $\vec p$. We say that an occurrence of a propositional variable $p$ in a formula $\phi$ is \emph{positive} (resp.\ \emph{negative}) if it is in the scope of an even (resp.\ odd) number of negations (here $\alpha\to\beta$ is regarded as $\neg\alpha\lor\beta$).

The semantics of the second-order propositional modal formulas are defined as follows:

\begin{definition}

A \emph{Kripke frame} is a pair $\mathbb{F}=(W,R)$ where $W\neq\emptyset$ is the \emph{domain} of $\mathbb{F}$, the \emph{accessibility relation} $R$ is a binary relation on $W$. A \emph{Kripke model} is a pair $\mathbb{M}=(\mathbb{F}, V)$ where $V:\mathsf{Prop}\to P(W)$ is a \emph{valuation} on $\mathbb{F}$. $V^{p}_{X}$ denote a valuation which is the same as $V$ except that $V^{p}_{X}(p)=X\subseteq W$.

Now the satisfaction relation can be defined as follows: given any Kripke model $\mathbb{M}=(W,R,V)$, any $w\in W$, the basic and Boolean cases are standard, and for modalities and propositional quantifiers, 

\begin{center}
\begin{tabular}{l c l}
$\mathbb{M},w\Vdash\Box\varphi$ & iff & for any $v$ such that $Rwv$, $\mathbb{M},v\Vdash\phi$;\\
$\mathbb{M},w\Vdash\Diamond\varphi$ & iff & there exists $v$ such that $Rwv\mbox{ and }\mathbb{M},v\Vdash\varphi$;\\
$\mathbb{M},w\Vdash\forall p\varphi$ & iff & for all $X\subseteq W$, $(W,R,V^{p}_{X}),w\Vdash\varphi$;\\
$\mathbb{M},w\Vdash\exists p\varphi$ & iff & there exists $X\subseteq W$ such that $(W,R,V^{p}_{X}),w\Vdash\varphi$.\\
\end{tabular}
\end{center}
\end{definition}

In order to use the algorithm to compute the first-order correspondents of Sahlqvist $\mathsf{SOPML}$ formulas, we will need the following \emph{expanded modal language} which is defined as follows\footnote{Notice that by adding the universal modality $\mathsf{A}$ into the language, all of the additional connectives in the expanded modal language can be defined in the language with $\mathsf{A}$. For example, $\mathsf{l}(\phi,\psi)$ can be rewritten as $\mathsf{A}(\phi\to\psi)$, and the backward-looking modality $\Diamondblack$ can be defined by $\Diamondblack\phi\leftrightarrow \exists p(p\land\forall q(q\to\mathsf{A}(p\to q)))\land \mathsf{E}(\phi\land\Diamond p)$ where $\mathsf{E}$ is $\neg\mathsf{A}\neg$. The expanded modal language is introduced for the convenience of the algorithm, as what is typically done in algorithmic correspondence theory.}:
$$\phi::=p\mid \nomi\mid \bot\mid \top\mid \neg\phi\mid \phi\land\phi\mid \phi\lor\phi\mid \phi\to\phi\mid$$
$$\Box\phi\mid\Diamond\phi\mid \blacksquare\phi\mid\Diamondblack\phi\mid \forall p\phi\mid \exists p\phi\mid \forall \nomi\phi\mid \exists \nomi\phi\mid \mathbf{l}(\phi,\phi)$$

where $p\in\mathsf{Prop}$, $\nomi\in\mathsf{Nom}$ is a \emph{nominal}, $\blacksquare$ and $\Diamondblack$ are the backward-looking box and diamond respectively, $\forall\nomi$ and $\exists\nomi$ are \emph{nominal quantifiers}, and $\mathbf{l}$ is a binary modality. We call a formula \emph{pure} if it does not contain propositional variables or propositional quantifiers (it can contain nominals, nominal quantifiers and the binary modality $\mathbf{l}$).

The interpretation of the expanded modal language is given as follows: For a valuation $V$, it is defined as $V:\mathsf{Prop}\cup\mathsf{Nom}\to P(W)$ such that $V(\nomi)$ is a singleton for all $\nomi\in\mathsf{Nom}$. The additional satisfaction clauses are given as follows (here $V^{\nomi}_{v}$ denote a valuation which is the same as $V$ except that $V^{\nomi}_{v}(\nomi)=\{v\}\subseteq W$.):

\begin{center}
\begin{tabular}{l c l}
$\mathbb{M},w\Vdash\nomi$ & iff & $V(\nomi)=\{w\}$;\\
$\mathbb{M},w\Vdash\blacksquare\varphi$ & iff & for any $v$ such that $Rvw$, $\mathbb{M},v\Vdash\phi$;\\
$\mathbb{M},w\Vdash\Diamondblack\varphi$ & iff & there exists $v$ such that $Rvw\mbox{ and }\mathbb{M},v\Vdash\varphi$;\\
$\mathbb{M},w\Vdash\forall \nomi\varphi$ & iff & for all $v\in W$, $(W,R,V^{\nomi}_{v}),w\Vdash\varphi$;\\
$\mathbb{M},w\Vdash\exists \nomi\varphi$ & iff & there exists $v\in W$ such that $(W,R,V^{\nomi}_{v}),w\Vdash\varphi$;\\
$\mathbb{M},w\Vdash\mathbf{l}(\phi,\psi)$ & iff & for all $v\in W$ (if $\mathbb{M},v\Vdash\phi$, then $\mathbb{M},v\Vdash\psi$).\\
\end{tabular}
\end{center}

We can extend $V$ to a map from the set of formulas to $P(W)$ in the natural way.

\subsection{Inequalities and complex inequalities}

We will find it convenient to use the inequality notation $\phi\leq\psi$ where $\phi$ and $\psi$ are formulas. We use $\mathsf{Ineq}$ to denote the set of all inequalities in the expanded modal language. We define \emph{complex inequalities} as follows:

$$\mathsf{Comp}::=\mathsf{Ineq}\mid\mathsf{Comp}\ \&\ \mathsf{Comp}\mid\mathsf{Comp}\ \Rightarrow\ \mathsf{Comp}\mid$$
$$\forall p\mathsf{Comp}\mid\exists p\mathsf{Comp}\mid\forall\nomi\mathsf{Comp}\mid\exists\nomi\mathsf{Comp}$$

Here we assume that the quantifiers have a higher precedence than $\&$, and $\&$ is higher than $\Rightarrow$.

Complex inequalities are interpreted in models $\mathbb{M}=(W,R,V)$ instead of pointed models $(\mathbb{M},w)$. The semantics of complex inequalities is defined as follows:

\begin{itemize}
\item An inequality is interpreted as follows:
$$(W,R,V)\Vdash\phi\leq\psi\mbox{ iff }$$$$(\mbox{for all }w\in W, \mbox{ if }(W,R,V),w\Vdash\phi, \mbox{ then }(W,R,V),w\Vdash\psi);$$

\item $(W,R,V)\Vdash \mathsf{Comp_{1}}\&\mathsf{Comp_{2}}$ iff $(W,R,V)\Vdash \mathsf{Comp_{1}}$ and $(W,R,V)\Vdash \mathsf{Comp_{2}}$;

\item $(W,R,V)\Vdash \mathsf{Comp_{1}}\Rightarrow\mathsf{Comp_{2}}$ iff ($(W,R,V)\Vdash \mathsf{Comp_{1}}$ implies $(W,R,V)\Vdash \mathsf{Comp_{2}}$);

\item $(W,R,V)\Vdash\forall p\mathsf{Comp}$ iff for all $X\subseteq W$, $(W,R,V^{p}_{X})\Vdash\mathsf{Comp}$;

\item $(W,R,V)\Vdash\exists p\mathsf{Comp}$ iff there exists an $X\subseteq W$ such that $(W,R,V^{p}_{X})\Vdash\mathsf{Comp}$;

\item $(W,R,V)\Vdash\forall\nomi\mathsf{Comp}$ iff for all $v\in W$, $(W,R,V^{\nomi}_{v})\Vdash\mathsf{Comp}$;

\item $(W,R,V)\Vdash\exists\nomi\mathsf{Comp}$ iff there exists an $v\in W$ such that $(W,R,V^{\nomi}_{v})\Vdash\mathsf{Comp}$.
\end{itemize}

\subsection{Standard translation}

In the correspondence language which is second-order due to the existence of propositional quantifiers in $\mathsf{SOPML}$, we have a binary predicate symbol $R$ corresponding to the binary relation, a set of constant symbols $i$ corresponding to each nominal $\nomi$, a set of unary predicate symbols $P$ corresponding to each propositional variable $p$.

\begin{definition}
The standard translation of the expanded $\mathsf{SOPML}$ language is defined as follows:

\begin{itemize}
\item $ST_{x}(p):=Px$;
\item $ST_{x}(\nomi):=x=i$;
\item $ST_{x}(\bot):=\bot$;
\item $ST_{x}(\top):=\top$;
\item $ST_{x}(\neg\phi):=\neg ST_{x}(\phi)$;
\item $ST_{x}(\phi\land\psi):=ST_{x}(\phi)\land ST_{x}(\psi)$;
\item $ST_{x}(\phi\lor\psi):=ST_{x}(\phi)\lor ST_{x}(\psi)$;
\item $ST_{x}(\phi\to\psi):=ST_{x}(\phi)\to ST_{x}(\psi)$;
\item $ST_{x}(\Box\phi):=\forall y(Rxy\to ST_{y}(\phi))$;
\item $ST_{x}(\Diamond\phi):=\exists y(Rxy\land ST_{y}(\phi))$;
\item $ST_{x}(\blacksquare\phi):=\forall y(Ryx\to ST_{y}(\phi))$;
\item $ST_{x}(\Diamondblack\phi):=\exists y(Ryx\land ST_{y}(\phi))$;
\item $ST_{x}(\forall p\phi):=\forall P ST_{x}(\phi)$;
\item $ST_{x}(\exists p\phi):=\exists P ST_{x}(\phi)$;
\item $ST_{x}(\forall \nomi\phi):=\forall i ST_{x}(\phi)$;
\item $ST_{x}(\exists \nomi\phi):=\exists i ST_{x}(\phi)$;
\item $ST_{x}(\mathbf{l}(\phi,\psi)):=\forall y(ST_{y}(\phi)\to ST_{y}(\psi))$.
\end{itemize}
\end{definition}

The following proposition states that this translation is correct:

\begin{proposition}
For any Kripke model $\mathbb{M}$, any $w\in W$ and any expanded $\mathsf{SOPML}$ formula $\phi$, 
$$\mathbb{M},w\Vdash\phi\mbox{ iff }\mathbb{M}\vDash ST_{x}(\phi)[x:=w].$$
\end{proposition}

For inequalities and complex inequalities, the standard translation is given in a global way:

\begin{definition}
\begin{itemize}
\item $ST(\phi\leq\psi):=\forall x(ST_{x}(\phi)\to ST_{x}(\psi))$;
\item $ST(\mathsf{Comp}_1\ \&\ \mathsf{Comp}_2)=ST(\mathsf{Comp}_1)\land ST(\mathsf{Comp}_2)$;
\item $ST(\mathsf{Comp}_1\ \Rightarrow\ \mathsf{Comp}_2)=ST(\mathsf{Comp}_1)\to ST(\mathsf{Comp}_2)$;
\item $ST(\forall p(\mathsf{Comp})):=\forall P(ST(\mathsf{Comp}))$;
\item $ST(\exists p(\mathsf{Comp})):=\exists P(ST(\mathsf{Comp}))$;
\item $ST(\forall \nomi(\mathsf{Comp})):=\forall i(ST(\mathsf{Comp}))$;
\item $ST(\exists \nomi(\mathsf{Comp})):=\exists i(ST(\mathsf{Comp}))$.
\end{itemize}
\end{definition}

\begin{proposition}\label{Prop:ST:ineq:quasi:mega}
For any Kripke model $\mathbb{M}$, any inequality $\mathsf{Ineq}$, any complex inequality $\mathsf{Comp}$,

$$\mathbb{M}\Vdash\mathsf{Ineq}\mbox{ iff }\mathbb{M}\vDash ST(\mathsf{Ineq});$$
$$\mathbb{M}\Vdash\mathsf{Comp}\mbox{ iff }\mathbb{M}\vDash ST(\mathsf{Comp}).$$
\end{proposition}

\section{Sahlqvist formulas in second-order propositional modal logic}\label{Sec:Sahl}

In this section, we define Sahlqvist formulas of second-order propositional modal logic step by step.

We first define (quantifier-free) positive formulas $\mathsf{POS}(\vec p)$ whose propositonal variables are among $\vec p$:
$$\mathsf{POS}(\vec p)::=p\mid \bot\mid \top \mid\mathsf{POS}(\vec p)\land\mathsf{POS}(\vec p) \mid\mathsf{POS}(\vec p)\lor\mathsf{POS}(\vec p)\mid \Box\mathsf{POS}(\vec p)\mid \Diamond\mathsf{POS}(\vec p)$$

where $p$ is in $\vec p$. These positive formulas have similar roles to the positive consequent part in Sahlqvist formulas in basic modal logic, which are going to receive minimal valuations. The reason why we do not allow propositional quantifiers in positive formulas is that we want the formula after receiving the minimal valuations to be translated into a first-order formula, while propositional quantifiers will make it second-order.

\subsection{The $\Pi_1$-fragment: Sahlqvist formulas in basic modal logic}

We define the $\Pi_1$-Sahlqvist antecedent $\mathsf{Sahl}_{1}(\vec p)$ whose propositonal variables are among $\vec p$:
$$\mathsf{Sahl}_{1}(\vec p)::=\Box^{n}p\mid \bot\mid \top\mid \neg\mathsf{POS}(\vec p)\mid\mathsf{Sahl}_{1}(\vec p)\land\mathsf{Sahl}_{1}(\vec p) \mid \Diamond\mathsf{Sahl}_{1}(\vec p)$$

where $p$ is in $\vec p$.

Then the $\Pi_1$-Sahlqvist formulas are defined as $\forall\vec p(\mathsf{Sahl}_{1}(\vec p)\to\mathsf{POS}(\vec p))$. Indeed, Sahlqvist formulas\footnote{In \cite[Chapter 3]{BRV01}, what we call Sahlqvist formulas are called Sahlqvist implications.} in the basic modal logic setting can be treated as universally quantified by propositional quantifiers which bind all occurrences of propositional variables, so in this sense the $\Pi_1$-Sahlqvist formulas can be taken as the Sahlqvist formulas in basic modal logic.

\subsection{The $\Pi_2$-fragment}

We define the $\mathsf{PIA}$ formula $\mathsf{PIA}(\vec q, \vec p)$ as follows:
$$\mathsf{PIA}(\vec q, \vec p)::=p\mid\Box\mathsf{PIA}(\vec q, \vec p)\mid \mathsf{PIA}(\vec q, \vec p)\land\mathsf{PIA}(\vec q, \vec p)\mid \mathsf{POS}(\vec q)\lor\mathsf{PIA}(\vec q, \vec p)$$

where $p$ is in $\vec p$. Here the $\mathsf{PIA}$ formula has two bunches of propositional variables: $\vec q$ is to receive minimal valuations for $\vec q$ from somewhere else, and $\vec p$ is used to compute minimal valuations for $\vec p$. Then it is easy to see that $\mathsf{PIA}(\vec q, \vec p)$ is equivalent to the form $\bigwedge\Box(\mathsf{POS}(\vec q)\lor\Box(\mathsf{POS}(\vec q)\lor\ldots p))$, where $p$ is in $\vec p$.

Now we can define $\Pi_2$-Sahlqvist antecedents as follows:
$$\mathsf{Sahl}_{2}(\vec p)::=\mathsf{Sahl}_{1}(\vec p)\mid
\forall\vec q(\mathsf{Sahl}_1(\vec q)\to\mathsf{PIA}(\vec q,\vec p))\mid\mathsf{Sahl}_{2}(\vec p)\land\mathsf{Sahl}_{2}(\vec p)\mid \Diamond\mathsf{Sahl}_{2}(\vec p)$$

Then $\Pi_2$-Sahlqvist formulas are defined as $\forall\vec p(\mathsf{Sahl}_{2}(\vec p)\to\mathsf{POS}(\vec p))$.

It is easy to see that formulas of the form $\forall\vec p(\mathsf{Sahl}_1(\vec p) \land \forall \vec q(\mathsf{Sahl}_1(\vec q) \to \mathsf{PIA}(\vec q,\vec p)) \to \mathsf{POS}(\vec p))$ are in the $\Pi_2$-hierarchy.

\subsection{The $\Pi_n$-fragment}

Now for the $\Pi_n$-fragment, assume that we have already defined $\Pi_{n-1}$-Sahlqvist antecedents $\mathsf{Sahl}_{n-1}(\vec p)$ and $\Pi_{n-1}$-Sahlqvist formulas $\forall\vec p(\mathsf{Sahl}_{n-1}(\vec p)\to\mathsf{POS}(\vec p))$, then we can define $\Pi_n$-Sahlqvist antecedents as follows:
$$\mathsf{Sahl}_{n}(\vec p)::=\mathsf{Sahl}_{n-1}(\vec p)\mid
\forall\vec q(\mathsf{Sahl}_{n-1}(\vec q)\to\mathsf{PIA}(\vec q,\vec p))\mid\mathsf{Sahl}_{n}(\vec p)\land\mathsf{Sahl}_{n}(\vec p)\mid \Diamond\mathsf{Sahl}_{n}(\vec p)$$

Then $\Pi_n$-Sahlqvist formulas are defined as $\forall\vec p(\mathsf{Sahl}_{n}(\vec p)\to\mathsf{POS}(\vec p))$.

\section{The Algorithm $\mathsf{ALBA}^{\mathsf{SOMPL}}$}\label{Sec:ALBA}

In the present section, we define the correspondence algorithm $\mathsf{ALBA}^{\mathsf{SOMPL}}$ for second-order propositional modal logic, in the style of \cite{CoGoVa06,CoPa12}. The algorithm receives a $\Pi_n$-Sahlqvist formula $\forall\vec p(\mathsf{Sahl}_{n}(\vec p)\to\mathsf{POS}(\vec p))$ as input and goes in three stages.

\begin{enumerate}

\item \textbf{Preprocessing and first approximation}:

The algorithm receives a $\Pi_n$-Sahlqvist formula $\forall\vec p(\mathsf{Sahl}_{n}(\vec p)\to\mathsf{POS}(\vec p))$ as input, and then apply the rewriting rule:

$$\infer{\forall\vec p(\mathsf{Sahl}_{n}(\vec p)\leq\mathsf{POS}(\vec p))}{\forall\vec p(\mathsf{Sahl}_{n}(\vec p)\to\mathsf{POS}(\vec p))}
$$

Then apply the first-approximation rule:

$$\infer{\forall\vec p\forall\nomi_0(\nomi_0\leq\mathsf{Sahl}_{n}(\vec p)\ \Rightarrow\ \nomi_0\leq\mathsf{POS}(\vec p))}{\forall\vec p(\mathsf{Sahl}_{n}(\vec p)\leq\mathsf{POS}(\vec p))}$$

\item \textbf{The reduction stage}:

In this stage, we aim at reducing $\nomi\leq\mathsf{Sahl}_{n}(\vec p)$ to a complex inequality in which $p$ occurs either in the form $\phi\leq p$ where $\phi$ is pure or in the form $\nomj\leq\neg\mathsf{POS}(\vec p)$.

\begin{enumerate}

\item The commutativity rule and the associativity rule for $\&$;

\item The rules for nominals:

\begin{enumerate}
\item Splitting rule:
$$
\infer[(Spl-Nom)]{\nomi\leq\alpha\ \&\ \nomi\leq\beta}{\nomi\leq\alpha\land\beta}
$$
\item Separation rule:
$$
\infer[(Sep-Nom)]{\nomi\leq\alpha\ \Rightarrow\ \nomi\leq\beta}{\nomi\leq\alpha\to\beta}
$$
\item Quantifier rule:
$$
\infer[(Quant-Nom)]{\forall q(\nomi\leq\alpha)}{\nomi\leq\forall q\alpha}
$$
\item Approximation rule:
$$
\infer[(Approx-Nom)]{\exists \nomj(\nomj\leq\alpha\ \&\ \nomi\leq\Diamond\nomj)}{\nomi\leq\Diamond\alpha}
$$
The nominals introduced by the approximation rule must not occur in the whole complex inequality before applying the rule.
\end{enumerate}

\item The residuation rules:

$$
\infer[(Res-\Box)]{\Diamondblack\alpha\leq\beta}{\alpha\leq\Box\beta}
\qquad
\infer[(Res-\lor)]{\alpha\land\neg\beta\leq\gamma}{\alpha\leq\beta\lor\gamma}
$$

\item The splitting rule:
$$
\infer[(Splitting)]{\alpha\leq\beta\ \&\ \alpha\leq\gamma}{\alpha\leq\beta\land\gamma}
$$

\item The quantifier rules:

$$
\infer[(Scope-\&)]{\exists\nomj(\mathsf{Comp}_{1}\ \&\ \mathsf{Comp}_{2})}{\exists\nomj(\mathsf{Comp}_{1})\ \&\ \mathsf{Comp}_{2}}
\qquad
\infer[(Scope-\Rightarrow)]{\forall\nomj(\mathsf{Comp}_{1}\ \Rightarrow\ \mathsf{Comp}_{2})}{\exists\nomj(\mathsf{Comp}_{1})\ \Rightarrow\ \mathsf{Comp}_{2}}
$$

where $\mathsf{Comp}_{2}$ does not have free occurrences of $\nomj$. 

$$
\infer[(Ex-pq)]{\forall p\forall q(\mathsf{Comp})}{\forall q\forall p(\mathsf{Comp})}
\qquad
\infer[(Ex-p\nomi)]{\forall p\forall\nomi(\mathsf{Comp})}{\forall\nomi\forall p(\mathsf{Comp})}
$$
$$
\infer[(Ex-\nomi p)]{\forall\nomi\forall p(\mathsf{Comp})}{\forall p\forall\nomi(\mathsf{Comp})}
\qquad
\infer[(Ex-\nomj\nomi)]{\forall\nomj\forall\nomi(\mathsf{Comp})}{\forall\nomi\forall\nomj(\mathsf{Comp})}
$$

$$
\infer[(Spl-Quant-p)]{\forall p(\mathsf{Comp}_{1}\Rightarrow\mathsf{Comp}_{2})\ \&\ \forall p(\mathsf{Comp}_{1}\Rightarrow\mathsf{Comp}_{3})}{\forall p(\mathsf{Comp}_{1}\Rightarrow(\mathsf{Comp}_{2}\&\mathsf{Comp}_{3}))}
$$

$$
\infer[(Spl-Quant-\nomi)]{\forall \nomi(\mathsf{Comp}_{1}\Rightarrow\mathsf{Comp}_{2})\ \&\ \forall\nomi(\mathsf{Comp}_{1}\Rightarrow\mathsf{Comp}_{3})}{\forall\nomi(\mathsf{Comp}_{1}\Rightarrow(\mathsf{Comp}_{2}\&\mathsf{Comp}_{3}))}
$$
\item The Ackermann rule:

In this step, we compute the minimal valuation for propositional variables and use the Ackermann rule to eliminate all the propositional variables.

$$
\infer{\alpha_{1}[\bigvee\psi/q]\leq\beta_{1}[\bigvee\psi/q]\&\ldots\&\alpha_{n}[\bigvee\psi/q]\leq\beta_{n}[\bigvee\psi/q]\Rightarrow\alpha[\bigvee\psi/q]\leq\beta[\bigvee\psi/q]}{\forall q(\alpha_{1}\leq\beta_{1}\&\ldots\&\alpha_{n}\leq\beta_{n}\ \&\ \psi_1\leq q\&\ldots\&\psi_m\leq q\Rightarrow \alpha\leq\beta)}
$$

where:

\begin{enumerate}
\item $\phi[\theta/p]$ means uniformly replace occurrences of $p$ in $\phi$ by $\theta$;
\item $\bigvee\psi=\psi_1\lor\ldots\lor\psi_m$;
\item Each $\alpha_i$ is positive, and each $\beta_i$ negative in $q$, for $1\leq i\leq n$;
\item $\alpha$ is negative in $q$ and $\beta$ is positive in $q$;
\item Each $\psi_i$ is pure (therefore $q$ does not occur in $\psi_{i}$).
\end{enumerate}

\item The packing rule:
$$\infer{\exists\nomi(\mathbf{l}(\alpha_{1},\beta_{1})\land\ldots\land\mathbf{l}(\alpha_{n},\beta_{n})\land\alpha)\leq\beta}{\forall \nomi(\alpha_1\leq\beta_1\&\ldots\&\alpha_n\leq\beta_n\Rightarrow\alpha\leq\beta)}
$$
where $\beta$ does not contain occurrences of $\nomi$.
\end{enumerate}

\item \textbf{Output}: By the execution of the algorithm, we can guarantee (see Theorem \ref{Thm:Success}) that given a $\Pi_n$-Sahlqvist formula as input, we can rewrite it into a pure complex inequality. Then we use standard translation to translate it into a first-order formula.
\end{enumerate}

From the design of the algorithm, we can see that it is specifically designed for $\Pi_n$-Sahlqvist formulas. Therefore, when we try to extend the $\Pi_n$-Sahlqvist fragment, we need to revise the rules accordingly.

\section{Soundness of $\mathsf{ALBA}^{\mathsf{SOPML}}$}\label{Sec:Soundness}

In the present section, we will prove the soundness of the algorithm $\mathsf{ALBA}^{\mathsf{SOPML}}$. The basic proof structure is similar to \cite{CoPa12}.

\begin{theorem}[Soundness]\label{Thm:Soundness}
If $\mathsf{ALBA}^{\mathsf{SOPML}}$ runs successfully on an input $\Pi_n$-Sahlqvist formula
$\forall\vec p(\mathsf{Sahl}_{n}(\vec p)\to\mathsf{POS}(\vec p))$ and outputs a first-order formula $\mathsf{FO}(\forall\vec p(\mathsf{Sahl}_{n}(\vec p)\to\mathsf{POS}(\vec p)))$, then for any Kripke frame $\mathbb{F}=(W,R)$, $$\mathbb{F}\Vdash\forall\vec p(\mathsf{Sahl}_{n}(\vec p)\to\mathsf{POS}(\vec p))\mbox{\ \ iff\ \ }\mathbb{F}\models\mathsf{FO}(\forall\vec p(\mathsf{Sahl}_{n}(\vec p)\to\mathsf{POS}(\vec p))).$$
\end{theorem}

\begin{proof}
The proof goes similarly to \cite[Theorem 8.1]{CoPa12}. Let $\forall\vec p(\mathsf{Sahl}_{n}(\vec p)\leq\mathsf{POS}(\vec p))$ denote the complex inequality after the first rewrite rule, $\forall\vec p\forall\nomi_0(\nomi_0\leq\mathsf{Sahl}_{n}(\vec p)\ \Rightarrow\ \nomi_0\leq\mathsf{POS}(\vec p))$ denote the complex inequality after the first approximation rule, $\mathsf{Comp}(\forall\vec p\forall\nomi_0(\nomi_0\leq\mathsf{Sahl}_{n}(\vec p)\ \Rightarrow\ \nomi_0\leq\mathsf{POS}(\vec p)))$ denote the complex inequality after Stage 2, and $\mathsf{FO}(\forall\vec p(\mathsf{Sahl}_{n}(\vec p)\to\mathsf{POS}(\vec p)))$ denote the standard translation of the complex inequality obtained after Stage 2, then it suffices to show the equivalence from (\ref{aCrct:Eqn0}) to (\ref{aCrct:Eqn4}) given below:
\begin{eqnarray}
&&\mathbb{F}\Vdash\forall\vec p(\mathsf{Sahl}_{n}(\vec p)\to\mathsf{POS}(\vec p))\label{aCrct:Eqn0}\\
&&\mathbb{F}\Vdash\forall\vec p(\mathsf{Sahl}_{n}(\vec p)\leq\mathsf{POS}(\vec p))\label{aCrct:Eqn1}\\
&&\mathbb{F}\Vdash\forall\vec p\forall\nomi_0(\nomi_0\leq\mathsf{Sahl}_{n}(\vec p)\ \Rightarrow\ \nomi_0\leq\mathsf{POS}(\vec p))\label{aCrct:Eqn2}\\
&&\mathbb{F}\Vdash\mathsf{Comp}(\forall\vec p\forall\nomi_0(\nomi_0\leq\mathsf{Sahl}_{n}(\vec p)\ \Rightarrow\ \nomi_0\leq\mathsf{POS}(\vec p)))\label{aCrct:Eqn3}\\
&&\mathbb{F}\vDash\mathsf{FO}(\forall\vec p(\mathsf{Sahl}_{n}(\vec p)\to\mathsf{POS}(\vec p)))\label{aCrct:Eqn4}
\end{eqnarray}

\begin{itemize}
\item the equivalence between (\ref{aCrct:Eqn0}) and (\ref{aCrct:Eqn1}) follows from Proposition \ref{prop:Soundness:stage:1:1};
\item the equivalence between (\ref{aCrct:Eqn1}) and (\ref{aCrct:Eqn2}) follows from Proposition \ref{prop:Soundness:stage:1:2};
\item the equivalence between (\ref{aCrct:Eqn2}) and (\ref{aCrct:Eqn3}) follows from Proposition \ref{Prop:Stage:2};
\item the equivalence between (\ref{aCrct:Eqn3}) and (\ref{aCrct:Eqn4}) follows from Proposition \ref{Prop:ST:ineq:quasi:mega}.
\end{itemize}
\end{proof}

In the remainder of this section, we prove the soundness of the rules in Stage 1 and 2.

\begin{proposition}[Soundness of the first rewrite rule in Stage 1]\label{prop:Soundness:stage:1:1}
The first rewrite rule is sound in both directions in $\mathbb{F}$, i.e.\ the formula before the rule is valid in $\mathbb{F}$ iff the complex inequality after the rule is valid in $\mathbb{F}$.
\end{proposition}

\begin{proof}
$\ $
\begin{center}
\begin{tabular}{r l}
& $\mathbb{F}\Vdash\forall\vec p(\mathsf{Sahl}_{n}(\vec p)\to\mathsf{POS}(\vec p))$\\
iff & for all $V$, ($\mathbb{F},V)\Vdash\forall\vec p(\mathsf{Sahl}_{n}(\vec p)\to\mathsf{POS}(\vec p))$\\
iff & for all $V$, for all $\vec X\subseteq W$, ($\mathbb{F},V^{\vec p}_{\vec X})\Vdash\mathsf{Sahl}_{n}(\vec p)\to\mathsf{POS}(\vec p)$\\
iff & for all $V$, for all $\vec X\subseteq W$, ($\mathbb{F},V^{\vec p}_{\vec X})\Vdash\mathsf{Sahl}_{n}(\vec p)\leq\mathsf{POS}(\vec p)$\\
iff & for all $V$, ($\mathbb{F},V)\Vdash\forall\vec p(\mathsf{Sahl}_{n}(\vec p)\leq\mathsf{POS}(\vec p))$\\
iff & $\mathbb{F}\Vdash\forall\vec p(\mathsf{Sahl}_{n}(\vec p)\leq\mathsf{POS}(\vec p))$.\\
\end{tabular}
\end{center}
\end{proof}

\begin{proposition}[Soundness of the first approximation rule in Stage 1]\label{prop:Soundness:stage:1:2}
The first approximation rule is sound in both directions in $\mathbb{F}$, i.e.\ the complex inequality before the rule is valid in $\mathbb{F}$ iff the complex inequality after the rule is valid in $\mathbb{F}$.
\end{proposition}

\begin{proof}
$\mathbb{F}\Vdash\forall\vec p(\mathsf{Sahl}_{n}(\vec p)\leq\mathsf{POS}(\vec p))$\\
iff for all $V$, ($\mathbb{F},V)\Vdash\forall\vec p(\mathsf{Sahl}_{n}(\vec p)\leq\mathsf{POS}(\vec p))$\\
iff for all $V$, for all $\vec X\subseteq W$, ($\mathbb{F},V^{\vec p}_{\vec X})\Vdash\mathsf{Sahl}_{n}(\vec p)\leq\mathsf{POS}(\vec p)$\\
iff for all $V$, all $\vec X\subseteq W$, all $w\in W$, ($\mathbb{F},V^{\vec p}_{\vec X}),w\Vdash\mathsf{Sahl}_{n}(\vec p)$ implies $(\mathbb{F},V^{\vec p}_{\vec X}),w\Vdash\mathsf{POS}(\vec p)$\\
iff for all $V$, all $\vec X\subseteq W$, all $w\in W$, ($\mathbb{F},V^{\vec p,\nomi_0}_{\vec X,w})\Vdash\nomi_0\leq\mathsf{Sahl}_{n}(\vec p)$ implies $(\mathbb{F},V^{\vec p,\nomi_0}_{\vec X,w})\Vdash\nomi_0\leq\mathsf{POS}(\vec p)$\\
iff for all $V$, ($\mathbb{F},V)\Vdash\forall \vec p\forall\nomi_0(\nomi_0\leq\mathsf{Sahl}_{n}(\vec p)\Rightarrow\nomi_0\leq\mathsf{POS}(\vec p))$\\
iff $\mathbb{F}\Vdash\forall \vec p\forall\nomi_0(\nomi_0\leq\mathsf{Sahl}_{n}(\vec p)\Rightarrow\nomi_0\leq\mathsf{POS}(\vec p))$.
\end{proof}

\begin{proposition}[Soundness of the rules in Stage 2]\label{Prop:Stage:2}
The rules in Stage 2 are sound in both directions in $\mathbb{F}$, i.e.\ the complex inequality before the rule is valid in $\mathbb{F}$ iff the complex inequality after the rule is valid in $\mathbb{F}$.
\end{proposition}

\begin{proof}
It suffices to show that each rule in Stage 2 is sound in both directions in $\mathbb{F}$.

\begin{itemize}
\item For the commutativity rule and associativity rule for $\&$, by the validity of $\alpha\land\beta\leftrightarrow\beta\land\alpha$ and $(\alpha\land\beta)\land\gamma\leftrightarrow\alpha\land(\beta\land\gamma)$.

\item For the splitting rule for nominals and the splitting rule for arbitrary formulas, it follows from the following equivalence: for all Kripke frame $\mathbb{F}$ and all $V$, $\mathbb{F},V\Vdash\alpha\leq\beta\land\gamma\mbox{ iff }(\mathbb{F},V\Vdash\alpha\leq\beta\mbox{ and }\mathbb{F},V\Vdash\alpha\leq\gamma)$.

\item For the separation rule for nominals, it follows from the following equivalence: for all $\mathbb{F}=(W,R)$ and all $V$,

$\mathbb{F},V\Vdash\nomi\leq\alpha\to\beta$\\
iff $\mathbb{F},V,V(\nomi)\Vdash\alpha\to\beta$\\
iff $\mathbb{F},V,V(\nomi)\Vdash\alpha$ implies $\mathbb{F},V,V(\nomi)\Vdash\beta$\\
iff $\mathbb{F},V\Vdash\nomi\leq\alpha$ implies $\mathbb{F},V\Vdash\nomi\leq\beta$\\
iff $\mathbb{F},V\Vdash\nomi\leq\alpha\Rightarrow\nomi\leq\beta$.

\item For the quantifier rule for nominals, it follows from the following equivalence: for all $\mathbb{F}=(W,R)$ and any $V$, 

$\mathbb{F},V\Vdash\nomi\leq\forall q\alpha$\\
iff $\mathbb{F},V,V(\nomi)\Vdash\forall q\alpha$\\
iff for all $X\subseteq W$, $\mathbb{F},V^{q}_{X},V(\nomi)\Vdash\alpha$\\
iff for all $X\subseteq W$, $\mathbb{F},V^{q}_{X},V^{q}_{X}(\nomi)\Vdash\alpha$\\
iff for all $X\subseteq W$, $\mathbb{F},V^{q}_{X}\Vdash\nomi\leq\alpha$\\
iff $\mathbb{F},V\Vdash\forall q(\nomi\leq\alpha)$.

\item For the approximation rule for nominals, it suffices to show that for any $\mathbb{F}=(W,R)$ and any $V$, 
\begin{enumerate}
\item if $(\mathbb{F},V)\Vdash\nomi\leq\Diamond\alpha$, then there is a valuation $V^{\nomj}$ such that $V^{\nomj}$ is the same as $V$ except $V^{\nomj}(\nomj)$, and $(\mathbb{F},V^{\nomj})\Vdash\nomi\leq\Diamond\nomj$ and $(\mathbb{F},V^{\nomj})\Vdash\nomj\leq\alpha$;
\item if $(\mathbb{F},V)\Vdash\nomi\leq\Diamond\nomj$ and $(\mathbb{F},V)\Vdash\nomj\leq\alpha$, then $(\mathbb{F},V)\Vdash\nomi\leq\Diamond\alpha$.
\end{enumerate}
For item 1, if $(\mathbb{F},V)\Vdash\nomi\leq\Diamond\alpha$, then $(\mathbb{F},V), V(\nomi)\Vdash\Diamond\alpha$, therefore there exists a $w\in W$ such that $(V(\nomi),w)\in R$ and $(\mathbb{F},V),w\Vdash\alpha$. Now take $V^{\nomj}$ such that $V^{\nomj}$ is the same as $V$ except that $V^{\nomj}(\nomj)=\{w\}$, then $(V^{\nomj}(\nomi), V^{\nomj}(\nomj))\in R$, so $(\mathbb{F},V^{\nomj})\Vdash\nomi\leq\Diamond\nomj$ and $(\mathbb{F},V^{\nomj})\Vdash\nomj\leq\alpha$.

For item 2, suppose $(\mathbb{F},V)\Vdash\nomi\leq\Diamond\nomj$ and $(\mathbb{F},V)\Vdash\nomj\leq\alpha$. Then $(V(\nomi), V(\nomj))\in R$ and $(\mathbb{F},V),V(\nomj)\Vdash\alpha$, so $(\mathbb{F},V),V(\nomi)\Vdash\Diamond\alpha$, therefore $(\mathbb{F},V)\Vdash\nomi\leq\Diamond\alpha$.

\item For the residuation rule for $\Box$, it suffices to show that for any $\mathbb{F}=(W,R)$ and any $V$, $(\mathbb{F},V)\Vdash\Diamondblack\alpha\leq\beta$ iff $(\mathbb{F},V)\Vdash\alpha\leq\Box\beta$.

$\Rightarrow$: if $(\mathbb{F},V)\Vdash\Diamondblack\alpha\leq\beta$, then for all $w\in W$, if $(\mathbb{F},V),w\Vdash\Diamondblack\alpha$, then $(\mathbb{F},V),w\Vdash\beta$. Our aim is to show that for all $v\in W$, if $(\mathbb{F},V),v\Vdash\alpha$, then $(\mathbb{F},V),v\Vdash\Box\beta$.

Consider any $v\in W$ such that $(\mathbb{F},V),v\Vdash\alpha$. Then for any $u\in W$ such that $(v,u)\in R$, $(\mathbb{F},V),u\Vdash\Diamondblack\alpha$. Since $(\mathbb{F},V)\Vdash\Diamondblack\alpha\leq\beta$, we have that $(\mathbb{F},V),u\Vdash\beta$, so for any $u\in W$ such that $(v,u)\in R$, $(\mathbb{F},V),u\Vdash\beta$, so $(\mathbb{F},V),v\Vdash\Box\beta$.

$\Leftarrow$: if $(\mathbb{F},V)\Vdash\alpha\leq\Box\beta$, then for all $w\in W$, if $(\mathbb{F},V),w\Vdash\alpha$, then $(\mathbb{F},V),w\Vdash\Box\beta$. Our aim is to show that for all $v\in W$, if $(\mathbb{F},V),v\Vdash\Diamondblack\alpha$, then $(\mathbb{F},V),v\Vdash\beta$.

Now assume that $(\mathbb{F},V),v\Vdash\Diamondblack\alpha$. Then there is a $u\in W$ such that $(u,v)\in R$ and $(\mathbb{F},V),u\Vdash\alpha$. By $(\mathbb{F},V)\Vdash\alpha\leq\Box\beta$, we have that $(\mathbb{F},V),u\Vdash\Box\beta$. Therefore, for $v\in W$, we have $(u,v)\in R$, thus $(\mathbb{F},V),v\Vdash\beta$.

\item For the residuation rule for $\lor$, it follows from the validity of $(\alpha\to(\beta\lor\gamma))\leftrightarrow((\alpha\land\neg\beta)\to\gamma)$.

\item For the quantifier scope rules, it follows from the validity of $\exists x\alpha\land\beta\leftrightarrow\exists x(\alpha\land\beta)$ and $(\exists x\alpha\to\beta)\leftrightarrow\forall x(\alpha\to\beta)$ (where $x$ does not occur in $\beta$).

\item For the quantifier exchange rules, it follows from the validity of $\forall P\forall Q\alpha\leftrightarrow\forall Q\forall P\alpha$, $\forall P\forall x\alpha\leftrightarrow\forall x\forall P\alpha$ and $\forall x\forall y\alpha\leftrightarrow\forall y\forall x\alpha$.

\item For the quantifier splitting rules, it follows from the validity of $\forall P(\alpha\to\beta\land\gamma)\leftrightarrow \forall P(\alpha\to\beta)\land\forall P(\alpha\to\gamma)$ and $\forall x(\alpha\to\beta\land\gamma)\leftrightarrow \forall x(\alpha\to\beta)\land\forall x(\alpha\to\gamma)$.

\item For the Ackermann rule, it suffices to show that for any $\mathbb{F}=(W,R)$ and any $V$, 

$\mathbb{F},V\Vdash\forall q(\alpha_{1}\leq\beta_{1}\&\ldots\&\alpha_{n}\leq\beta_{n}\ \&\ \psi_1\leq q\&\ldots\&\psi_m\leq q\Rightarrow \alpha\leq\beta)$\\
iff $\mathbb{F},V\Vdash\alpha_{1}[\bigvee\psi/q]\leq\beta_{1}[\bigvee\psi/q]\&\ldots\&\alpha_{n}[\bigvee\psi/q]\leq\beta_{n}[\bigvee\psi/q]\Rightarrow\alpha[\bigvee\psi/q]\leq\beta[\bigvee\psi/q]$.

$\Rightarrow$: Easy, by instantiation of the propositional quantifier.

$\Leftarrow$: Assmue $\mathbb{F},V\Vdash\alpha_{1}[\bigvee\psi/q]\leq\beta_{1}[\bigvee\psi/q]\&\ldots\&\alpha_{n}[\bigvee\psi/q]\leq\beta_{n}[\bigvee\psi/q]\Rightarrow\alpha[\bigvee\psi/q]\leq\beta[\bigvee\psi/q]$. Then for any $X\subseteq W$, it suffices to show that if $\mathbb{F},V^{q}_{X}\Vdash\alpha_{1}\leq\beta_{1}\&\ldots\&\alpha_{n}\leq\beta_{n}\ \&\ \psi_1\leq q\&\ldots\&\psi_m\leq q$, then $\mathbb{F},V^{q}_{X}\Vdash\alpha\leq\beta$. Now assume $\mathbb{F},V^{q}_{X}\Vdash\alpha_{1}\leq\beta_{1}\&\ldots\&\alpha_{n}\leq\beta_{n}\ \&\ \psi_1\leq q\&\ldots\&\psi_m\leq q$, then $V^{q}_{X}(\alpha_i)\subseteq V^{q}_{X}(\beta_i)$ for $1\leq i\leq n$, and $V^{q}_{X}(\psi_j)\subseteq X$ for $1\leq j\leq m$, therefore $V(\bigvee\psi)=V^{q}_{X}(\bigvee\psi)\subseteq X$. Since each $\alpha_i$ is positive and each $\beta_i$ is negative in $q$, we have that $V(\alpha_i[\bigvee\psi/q])\subseteq V^{q}_{X}(\alpha_i)\subseteq V^{q}_{X}(\beta_i)\subseteq V(\beta_i[\bigvee\psi/q])$, $1\leq i\leq n$, so by $\mathbb{F},V\Vdash\alpha_{1}[\bigvee\psi/q]\leq\beta_{1}[\bigvee\psi/q]\&\ldots\&\alpha_{n}[\bigvee\psi/q]\leq\beta_{n}[\bigvee\psi/q]\Rightarrow\alpha[\bigvee\psi/q]\leq\beta[\bigvee\psi/q]$ we have $V(\alpha[\bigvee\psi/q])\subseteq V(\beta[\bigvee\psi/q])$, therefore by $\alpha$ is negative and $\beta$ is positive in $q$, we have $V^{q}_{X}(\alpha)\subseteq V(\alpha[\bigvee\psi/q])\subseteq V(\beta[\bigvee\psi/q])\subseteq V^{q}_{X}(\beta)$, so $\mathbb{F},V^{q}_{X}\Vdash\alpha\leq\beta$, which concludes the proof.

\item For the packing rule, it follows from the following equivalence: for any $\mathbb{F}=(W,R)$ and any $V$, 

$\mathbb{F},V\Vdash\forall \nomi(\alpha_1\leq\beta_1\&\ldots\&\alpha_n\leq\beta_n\Rightarrow\alpha\leq\beta)$\\
iff for all $w\in W$, $(\mathbb{F},V^{\nomi}_{w})\Vdash\alpha_1\leq\beta_1\&\ldots\&\alpha_n\leq\beta_n\Rightarrow\alpha\leq\beta$\\
iff for all $w\in W$, if $(\mathbb{F},V^{\nomi}_{w})\Vdash\alpha_i\leq\beta_i$ for $1\leq i\leq n$, then $(\mathbb{F},V^{\nomi}_{w})\Vdash\alpha\leq\beta$\\
iff for all $w\in W$, if $(\mathbb{F},V^{\nomi}_{w})\Vdash\mathbf{l}(\alpha_i,\beta_i)$ for $1\leq i\leq n$, then $(\mathbb{F},V^{\nomi}_{w})\Vdash\alpha\leq\beta$\\
iff for all $w\in W$, if $(\mathbb{F},V^{\nomi}_{w})\Vdash\mathbf{l}(\alpha_i,\beta_i)$ for $1\leq i\leq n$, then for all $v\in W$, if $(\mathbb{F},V^{\nomi}_{w}),v\Vdash\alpha$ then $(\mathbb{F},V^{\nomi}_{w}),v\Vdash\beta$\\
iff for all $w,v\in W$, if $(\mathbb{F},V^{\nomi}_{w})\Vdash\mathbf{l}(\alpha_i,\beta_i)$ for $1\leq i\leq n$, then if $(\mathbb{F},V^{\nomi}_{w}),v\Vdash\alpha$ then $(\mathbb{F},V^{\nomi}_{w}),v\Vdash\beta$\\
iff for all $w,v\in W$, if $(\mathbb{F},V^{\nomi}_{w})\Vdash\mathbf{l}(\alpha_i,\beta_i)$ for $1\leq i\leq n$ and $(\mathbb{F},V^{\nomi}_{w}),v\Vdash\alpha$, then $(\mathbb{F},V^{\nomi}_{w}),v\Vdash\beta$\\
iff for all $w,v\in W$, if $(\mathbb{F},V^{\nomi}_{w}),v\Vdash\mathbf{l}(\alpha_i,\beta_i)$ for $1\leq i\leq n$ and $(\mathbb{F},V^{\nomi}_{w}),v\Vdash\alpha$, then $(\mathbb{F},V^{\nomi}_{w}),v\Vdash\beta$\\
iff for all $w,v\in W$, if $(\mathbb{F},V^{\nomi}_{w}),v\Vdash\mathbf{l}(\alpha_1,\beta_1)\land\ldots\land\mathbf{l}(\alpha_n,\beta_n)\land\alpha$, then $(\mathbb{F},V^{\nomi}_{w}),v\Vdash\beta$\\
iff for all $w,v\in W$, if $(\mathbb{F},V^{\nomi}_{w}),v\Vdash\mathbf{l}(\alpha_1,\beta_1)\land\ldots\land\mathbf{l}(\alpha_n,\beta_n)\land\alpha$, then $(\mathbb{F},V),v\Vdash\beta$\\
iff for all $v\in W$, if there exists a $w\in W$ such that $(\mathbb{F},V^{\nomi}_{w}),v\Vdash\mathbf{l}(\alpha_1,\beta_1)\land\ldots\land\mathbf{l}(\alpha_n,\beta_n)\land\alpha$, then $(\mathbb{F},V),v\Vdash\beta$\\
iff for all $v\in W$, if $(\mathbb{F},V),v\Vdash\exists\nomi(\mathbf{l}(\alpha_1,\beta_1)\land\ldots\land\mathbf{l}(\alpha_n,\beta_n)\land\alpha)$, then $(\mathbb{F},V),v\Vdash\beta$\\
iff $\mathbb{F},V\Vdash\exists\nomi(\mathbf{l}(\alpha_{1},\beta_{1})\land\ldots\land\mathbf{l}(\alpha_{n},\beta_{n})\land\alpha)\leq\beta$.
\end{itemize}
\end{proof}

\section{Success of $\mathsf{ALBA}^{\mathsf{SOPML}}$ on $\Pi_n$-Sahlqvist formulas}\label{Sec:Success}

By success of $\mathsf{ALBA}^{\mathsf{SOPML}}$ on $\Pi_n$-Sahlqvist formulas we mean that the algorithm $\mathsf{ALBA}^{\mathsf{SOPML}}$ can transform any input $\Pi_n$-Sahlqvist formula into a pure complex inequality which does not contain any propositional variables or any propositional quantifiers (here we allow nominal quantifiers to occur). We prove this by induction on $n$ that $\mathsf{ALBA}^{\mathsf{SOPML}}$ successfully transforms $\nomi\leq\mathsf{Sahl}_{n}(\vec p)$ into given shapes.

\begin{proposition}\label{Prop:Sahl_1}
In the reduction stage, by running the algorithm $\mathsf{ALBA}^{\mathsf{SOPML}}$, $\nomi\leq\mathsf{Sahl}_{1}(\vec p)$ can be reduced to the following complex inequality:

$$\exists\vec \nomj(\mathsf{NEG}\ \&\ \mathsf{NOM}\ \&\ \mathsf{MinVal})$$

where
\begin{itemize}
\item $\exists\vec \nomj$ is a (possibly empty) bunch of nominal quantifiers;
\item $\mathsf{NEG}$ is a (possibly empty) meta-conjunction of inequalities of the form $\nomj\leq\neg\mathsf{POS}(\vec p)$, where $\nomj$ is either $\nomi$ or in $\vec \nomj$,
\item $\mathsf{NOM}$ is a (possibly empty) meta-conjunction of inequalities of the form $\nomj\leq\Diamond\nomk$, where $\nomj,\nomk$ are either $\nomi$ or in $\vec \nomj$,
\item $\mathsf{MinVal}$ is a (possibly empty) meta-conjunction of inequalities of the form $\psi\leq p$, where $\psi$ is pure and $p$ is in $\vec p$.
\end{itemize}
\end{proposition}

\begin{proof}
We prove by induction on the formula complexity of $\mathsf{Sahl}_{1}(\vec p)$. 

\begin{itemize}
\item For the case where $\mathsf{Sahl}_{1}(\vec p)=\bot,\top$, trivial.
\item For the case where $\mathsf{Sahl}_{1}(\vec p)=\Box^{n}p$, by applying the residuation rule for $\Box$, we get $\Diamondblack^{n}\nomi\leq p$, which belongs to $\mathsf{MinVal}$.
\item For the case where $\mathsf{Sahl}_{1}(\vec p)=\neg\mathsf{POS}(\vec p)$, it already belongs to $\mathsf{NEG}$.
\item For the case where $\mathsf{Sahl}_{1}(\vec p)=\mathsf{Sahl}^{a}_{1}(\vec p)\land\mathsf{Sahl}^{b}_{1}(\vec p)$, we first apply (Spl-Nom) to $\nomi\leq\mathsf{Sahl}_{1}(\vec p)$ and get $\nomi\leq\mathsf{Sahl}^{a}_{1}(\vec p)$ and $\nomi\leq\mathsf{Sahl}^{b}_{1}(\vec p)$. Then we apply the induction hypothesis and get $$\exists\vec \nomj^{a}(\mathsf{NEG}^{a}\ \&\ \mathsf{NOM}^{a}\ \&\ \mathsf{MinVal}^{a})\ \&\ \exists\vec \nomj^{b}(\mathsf{NEG}^{b}\ \&\ \mathsf{NOM}^{b}\ \&\ \mathsf{MinVal}^{b}).$$ By applying the (Scope-$\&$) rule and commutativity and associativity rules for $\&$, we get the desired shape.
\item For the case where $\mathsf{Sahl}_{1}(\vec p)=\Diamond\mathsf{Sahl}^{a}_{1}(\vec p)$, we first apply (Approx-Nom) for $\Diamond$ and get $\exists\nomk(\nomk\leq\mathsf{Sahl}^{a}_{1}(\vec p)\ \&\ \nomi\leq\Diamond\nomk)$. Then we apply the induction hypothesis to $\nomk\leq\mathsf{Sahl}^{a}_{1}(\vec p)$ and get $$\exists\nomk(\exists\vec \nomj(\mathsf{NEG}\ \&\ \mathsf{NOM}\ \&\ \mathsf{MinVal})\ \&\ \nomi\leq\Diamond\nomk).$$ By applying the (Scope-$\&$) rule and commutativity and associativity rules for $\&$, we get the desired shape ($\nomi\leq\Diamond\nomk$ is merged into $\mathsf{NOM}$).
\end{itemize}
\end{proof}

\begin{proposition}\label{Prop:PIA}
In the reduction stage, by running the algorithm $\mathsf{ALBA}^{\mathsf{SOPML}}$, for any formula $\psi$ such that 
\begin{itemize}
\item $\psi$ contains no propositional quantifiers;
\item $\psi$ contains propositional variables at most from $\vec q$;
\item all occurrences of $\vec q$-variables are negative in $\psi$; 
\end{itemize}
$\psi\leq\mathsf{PIA}(\vec q,\vec p)$ can be reduced to the following complex inequality:
$$\mathsf{RelMinVal}(\vec q, \vec p))$$
where $\mathsf{RelMinVal}(\vec q,\vec p)$\footnote{Here $\mathsf{RelMinVal}$ means relative minimal valuation.} is a meta-conjunction of inequalities of the form $\phi\leq p$, $\phi$ has the three properties for $\psi$ stated above, and $p$ is in $\vec p$.

Especially, this proposition holds for $\psi=\nomi$.
\end{proposition}

\begin{proof}
We prove by induction on the complexity of $\mathsf{PIA}(\vec q,\vec p)$.
\begin{itemize}
\item For the basic case where $\mathsf{PIA}(\vec q,\vec p)=p$, trivial.
\item For the case where $\mathsf{PIA}(\vec q,\vec p)=\Box\mathsf{PIA}^{a}(\vec q,\vec p)$, we first apply the (Res-$\Box$) rule and get $\Diamondblack\psi\leq\mathsf{PIA}^{a}(\vec q,\vec p)$. Then by induction hypothesis, it is transformed into $\mathsf{RelMinVal}(\vec q, \vec p))$ of the required shape.
\item For the case where $\mathsf{PIA}(\vec q,\vec p)=\mathsf{PIA}^{a}(\vec q,\vec p)\land\mathsf{PIA}^{b}(\vec q,\vec p)$, we first apply (Splitting) and get $\psi\leq\mathsf{PIA}^{a}(\vec q,\vec p)$ and $\psi\leq\mathsf{PIA}^{b}(\vec q,\vec p)$. Then by induction hypothesis, these two inequalities can be transformed into $\mathsf{RelMinVal}^{a}(\vec q, \vec p))$ and $\mathsf{RelMinVal}^{b}(\vec q, \vec p))$ of the required shape, which put together is also of the required shape.
\item For the case where $\mathsf{PIA}(\vec q,\vec p)=\mathsf{POS}(\vec q)\lor\mathsf{PIA}^{a}(\vec q,\vec p)$, by applying (Res-$\lor$), we get $\psi\land\neg\mathsf{POS}(\vec q)\leq \mathsf{PIA}^{a}(\vec q,\vec p)$. Then $\psi\land\neg\mathsf{POS}(\vec q)$ satisfies the conditions required in the proposition, so we can apply the induction hypothesis and get the $\mathsf{RelMinVal}(\vec q, \vec p))$ of the required shape.
\end{itemize}
\end{proof}

\begin{proposition}\label{Prop:Sahl_1_to_PIA}
In the reduction stage, by running the algorithm $\mathsf{ALBA}^{\mathsf{SOPML}}$, $\nomi\leq\forall\vec q(\mathsf{Sahl}_1(\vec q)\to\mathsf{PIA}(\vec q,\vec p))$ can be reduced to the following complex inequality:
$$\forall\vec \nomj(\mathsf{PURE}\ \Rightarrow\ \mathsf{MinVal}(\vec p))$$
where 
\begin{itemize}
\item $\mathsf{PURE}$ is a meta-conjunction of pure inequalities,
\item $\mathsf{MinVal}(\vec p)$ is a meta-conjunction of inequalities of the form $\psi\leq p$, where $\psi$ is pure and $p$ is in $\vec p$.
\end{itemize}
Therefore, $\nomi\leq\forall\vec q(\mathsf{Sahl}_1(\vec q)\to\mathsf{PIA}(\vec q,\vec p))$ can be reduced to the form 
$$\mathsf{MinVal}(\vec p),$$
where $\mathsf{MinVal}(\vec p)$ is a meta-conjunction of inequalities of the form $\psi\leq p$, where $\psi$ is pure and $p$ is in $\vec p$.
\end{proposition}

\begin{proof}
We first apply (Quant-Nom) on $$\nomi\leq\forall\vec q(\mathsf{Sahl}_1(\vec q)\to\mathsf{PIA}(\vec q,\vec p)),$$ then apply (Sep-Nom) we get $$\forall\vec q(\nomi\leq\mathsf{Sahl}_1(\vec q)\Rightarrow\nomi\leq\mathsf{PIA}(\vec q,\vec p)).$$ By Proposition \ref{Prop:Sahl_1}, we have 

$$\forall\vec q(\exists\vec \nomj(\mathsf{NEG}\ \&\ \mathsf{NOM}\ \&\ \mathsf{MinVal})\Rightarrow\nomi\leq\mathsf{PIA}(\vec q,\vec p)).$$

By Proposition \ref{Prop:PIA}, we have 

$$\forall\vec q(\exists\vec \nomj(\mathsf{NEG}\ \&\ \mathsf{NOM}\ \&\ \mathsf{MinVal})\Rightarrow\mathsf{RelMinVal}(\vec q, \vec p)).$$

Then by applying (Scope-$\Rightarrow$) and repeatedly applying (Ex-$\nomi p$), we have 

$$\forall\vec \nomj\forall\vec q(\mathsf{NEG}\ \&\ \mathsf{NOM}\ \&\ \mathsf{MinVal}\Rightarrow\mathsf{RelMinVal}(\vec q, \vec p)).$$

Then by applying the Ackermann rule for each propositional variable in $\vec q$, $\mathsf{NEG}$ receives the minimal valuation from $\mathsf{MinVal}$ and become a meta-conjunction of pure inequalities, $\mathsf{NOM}$ remains pure, $\mathsf{MinVal}$ disappears, and $\mathsf{RelMinVal}(\vec q, \vec p))$ becomes a meta-conjunction of inequalities of the form $\psi\leq p$ where $\psi$ is pure and $p$ is in $\vec p$. Now what we have is the following shape, as required by the proposition:
$$\forall\vec \nomj(\mathsf{PURE}\Rightarrow\mathsf{MinVal}(\vec p)).$$
Then apply (Spl-Quant-$\nomi$) and the packing rule, one get a complex inequality $\mathsf{MinVal}(\vec p)$ of the required form.

\end{proof}

\begin{proposition}\label{Prop:Sahl_2}
In the reduction stage, by running the algorithm $\mathsf{ALBA}^{\mathsf{SOPML}}$, $\nomi\leq\mathsf{Sahl}_{2}(\vec p)$ can be reduced to the following complex inequality:
$$\exists\vec \nomj(\mathsf{NEG}\ \&\ \mathsf{NOM}\ \&\ \mathsf{MinVal})$$
where $\exists\vec \nomj$, $\mathsf{NEG}$, $\mathsf{NOM}$, $\mathsf{MinVal}$ are described as in Proposition \ref{Prop:Sahl_1}.
\end{proposition}

\begin{proof}
We prove by induction on the complexity of $\mathsf{Sahl}_{2}(\vec p)$.
\begin{itemize}
\item For the case where $\mathsf{Sahl}_{2}(\vec p)=\mathsf{Sahl}_{1}(\vec p)$, see Proposition \ref{Prop:Sahl_1}.
\item For the case where $\mathsf{Sahl}_{2}(\vec p)=\forall\vec q(\mathsf{Sahl}_1(\vec q)\to\mathsf{PIA}(\vec q,\vec p))$, by Proposition \ref{Prop:Sahl_1_to_PIA}, $\nomi\leq\forall\vec q(\mathsf{Sahl}_1(\vec q)\to\mathsf{PIA}(\vec q,\vec p))$ is reduced to $\forall\vec \nomj(\mathsf{PURE}\ \Rightarrow\ \mathsf{MinVal}(\vec p))$. Now apply (Spl-Quant-$\nomi$) and the packing rule, we have a meta-conjunction of inequalities of the form $\phi\leq p$ where $\phi$ is pure and $p$ is in $\vec p$, so it belongs to $\mathsf{MinVal}$.
\item For the case where $\mathsf{Sahl}_{2}(\vec p)=\mathsf{Sahl}^{a}_{2}(\vec p)\land\mathsf{Sahl}^{b}_{2}(\vec p)$, similar to the $\mathsf{Sahl}_{1}(\vec p)=\mathsf{Sahl}^{a}_{1}(\vec p)\land\mathsf{Sahl}^{b}_{1}(\vec p)$ case in the proof of Proposition \ref{Prop:Sahl_1}.
\item For the case where $\mathsf{Sahl}_{2}(\vec p)=\Diamond\mathsf{Sahl}^{a}_{2}(\vec p)$, similar to the $\mathsf{Sahl}_{1}(\vec p)=\Diamond\mathsf{Sahl}^{a}_{1}(\vec p)$ case in the proof of Proposition \ref{Prop:Sahl_1}.
\end{itemize}
\end{proof}

\begin{proposition}\label{Prop:Sahl_n}
In the reduction stage, by running the algorithm $\mathsf{ALBA}^{\mathsf{SOPML}}$, $\nomi\leq\mathsf{Sahl}_{n}(\vec p)$ can be reduced to the following complex inequality:
$$\exists\vec \nomj(\mathsf{NEG}\ \&\ \mathsf{NOM}\ \&\ \mathsf{MinVal})$$
where $\exists\vec \nomj$, $\mathsf{NEG}$, $\mathsf{NOM}$, $\mathsf{MinVal}$ are described as in Proposition \ref{Prop:Sahl_1}.
\end{proposition}

\begin{proof}
We prove by induction on $n$. For $n=1,2$, they are already proved in Proposition \ref{Prop:Sahl_1} and \ref{Prop:Sahl_2}. Now we assume that for $n=k$ the property holds, then by an argument similar to Proposition \ref{Prop:Sahl_1_to_PIA}, we have that $\nomi\leq\forall\vec q(\mathsf{Sahl}_k(\vec q)\to\mathsf{PIA}(\vec q,\vec p))$ can be reduced to the following complex inequality:
$$\forall\vec \nomj(\mathsf{PURE}\ \Rightarrow\ \mathsf{MinVal}(\vec p))$$
where $\mathsf{PURE}$ and $\mathsf{MinVal}(\vec p)$ are as described in Proposition \ref{Prop:Sahl_1_to_PIA}. Then by an argument similar to Proposition \ref{Prop:Sahl_2}, $\nomi\leq\mathsf{Sahl}_{k+1}(\vec p)$ can be reduced to the following complex inequality:
$$\exists\vec \nomj(\mathsf{NEG}\ \&\ \mathsf{NOM}\ \&\ \mathsf{MinVal})$$
where $\exists\vec \nomj$, $\mathsf{NEG}$, $\mathsf{NOM}$, $\mathsf{MinVal}$ are described as in Proposition \ref{Prop:Sahl_1}, hence the property holds for $n=k+1$.
\end{proof}

\begin{theorem}\label{Thm:Success}
For any $\Pi_n$-Sahlqvist formula, the algorithm $\mathsf{ALBA}^{\mathsf{SOPML}}$ transforms it into a complex inequality which does not contain any occurrences of propositional variables or propositional quantifiers.
\end{theorem}

\begin{proof}
Given a $\Pi_n$-Sahlqvist formula $\forall\vec p(\mathsf{Sahl}_{n}(\vec p)\to\mathsf{POS}(\vec p))$, we first apply the rules in Stage 1 and get $$\forall\vec p\forall\nomi_0(\nomi_0\leq\mathsf{Sahl}_{n}(\vec p)\ \Rightarrow\ \nomi_0\leq\mathsf{POS}(\vec p)).$$ By Proposition \ref{Prop:Sahl_n}, we have 
$$\forall\vec p\forall\nomi_0(\exists\vec \nomj(\mathsf{NEG}\ \&\ \mathsf{NOM}\ \&\ \mathsf{MinVal})\Rightarrow\nomi_0\leq\mathsf{POS}(\vec p)).$$
Then by applying (Scope-$\Rightarrow$) and repeatedly applying (Ex-$\nomi p$), we have 
$$\forall\nomi_0\forall\vec \nomj\forall\vec p(\mathsf{NEG}\ \&\ \mathsf{NOM}\ \&\ \mathsf{MinVal}\Rightarrow\nomi_0\leq\mathsf{POS}(\vec p)).$$
Now we can apply the Ackermann rule repeatedly for each propositional variable $p$ in $\vec p$, then $\mathsf{NEG}$ receives the minimal valuation from $\mathsf{MinVal}$ and become a meta-conjunction of pure inequalities, $\mathsf{NOM}$ remains pure, $\mathsf{MinVal}$ disappears, and $\nomi_0\leq\mathsf{POS}(\vec p)$ receives the minimal valuation and becomes pure. Now what we have is the following shape:
$$\forall\nomi_0\forall\vec \nomj(\mathsf{PURE}\Rightarrow \mathsf{PURE}'),$$
where $\mathsf{PURE}$ is a meta-conjunction of pure inequalities, and $\mathsf{PURE}'$ is a pure inequality.
\end{proof}

\begin{corollary}
There is an algorithm such that for any $\Pi_n$-Sahlqvist formula $\phi$, it can be transformed into an equivalent first-order formula.
\end{corollary}

\section{Examples, non-standard rules and canonicity}\label{Sec:Example}

\subsection{Examples}

We give three examples of $\Pi_2$-Sahlqvist formulas to show how the $\mathsf{ALBA}^{\mathsf{SOPML}}$ algorithm works:

\begin{example}
$\forall p(\Diamond\Box p\land \forall q(\Diamond\Box q\to\Box(\Box q\lor\Box p))\to \Box\Diamond\Box p)$\\
$\forall p\forall\nomi(\nomi\leq\Diamond\Box p\land \forall q(\Diamond\Box q\to\Box(\Box q\lor\Box p))\Rightarrow\nomi\leq\Box\Diamond\Box p)$\\
$\forall p\forall\nomi(\nomi\leq\Diamond\Box p\ \&\ \nomi\leq\forall q(\Diamond\Box q\to\Box(\Box q\lor\Box p))\Rightarrow\nomi\leq\Box\Diamond\Box p)$\\
$\forall p\forall\nomi(\nomi\leq\Diamond\Box p\ \&\ \forall q(\nomi\leq\Diamond\Box q\to\Box(\Box q\lor\Box p))\Rightarrow\nomi\leq\Box\Diamond\Box p)$\\
$\forall p\forall\nomi(\nomi\leq\Diamond\Box p\ \&\ \forall q(\nomi\leq\Diamond\Box q\Rightarrow\nomi\leq\Box(\Box q\lor\Box p))\Rightarrow\nomi\leq\Box\Diamond\Box p)$\\
$\forall p\forall\nomi(\nomi\leq\Diamond\Box p\ \&\ \forall q\forall\nomj(\nomi\leq\Diamond\nomj\ \&\ \nomj\leq\Box q\Rightarrow\nomi\leq\Box(\Box q\lor\Box p))\Rightarrow\nomi\leq\Box\Diamond\Box p)$\\
$\forall p\forall\nomi(\nomi\leq\Diamond\Box p\ \&\ \forall q\forall\nomj(\nomi\leq\Diamond\nomj\ \&\ \Diamondblack\nomj\leq q\Rightarrow\nomi\leq\Box(\Box q\lor\Box p))\Rightarrow\nomi\leq\Box\Diamond\Box p)$\\
$\forall p\forall\nomi(\nomi\leq\Diamond\Box p\ \&\ \forall\nomj(\nomi\leq\Diamond\nomj \Rightarrow\nomi\leq\Box(\Box \Diamondblack\nomj\lor\Box p))\Rightarrow\nomi\leq\Box\Diamond\Box p)$\\
$\forall p\forall\nomi(\nomi\leq\Diamond\Box p\ \&\ \forall\nomj(\nomi\leq\Diamond\nomj \Rightarrow\Diamondblack\nomi\leq\Box\Diamondblack\nomj\lor\Box p)\Rightarrow\nomi\leq\Box\Diamond\Box p)$\\
$\forall p\forall\nomi(\nomi\leq\Diamond\Box p\ \&\ \forall\nomj(\nomi\leq\Diamond\nomj \Rightarrow\Diamondblack\nomi\land\neg\Box\Diamondblack\nomj\leq\Box p)\Rightarrow\nomi\leq\Box\Diamond\Box p)$\\
$\forall p\forall\nomi(\nomi\leq\Diamond\Box p\ \&\ \forall\nomj(\nomi\leq\Diamond\nomj \Rightarrow\Diamondblack(\Diamondblack\nomi\land\neg\Box\Diamondblack\nomj)\leq p)\Rightarrow\nomi\leq\Box\Diamond\Box p)$\\
$\forall p\forall\nomi(\nomi\leq\Diamond\Box p\ \&\ \forall\nomj(\mathbf{l}(\nomi,\Diamond\nomj)\land\Diamondblack(\Diamondblack\nomi\land\neg\Box\Diamondblack\nomj)\leq p)\Rightarrow\nomi\leq\Box\Diamond\Box p)$\\
$\forall p\forall\nomi(\nomi\leq\Diamond\Box p\ \&\ \exists\nomj(\mathbf{l}(\nomi,\Diamond\nomj)\land\Diamondblack(\Diamondblack\nomi\land\neg\Box\Diamondblack\nomj))\leq p\Rightarrow\nomi\leq\Box\Diamond\Box p)$\\

now denote $\exists\nomj(\mathbf{l}(\nomi,\Diamond\nomj)\land\Diamondblack(\Diamondblack\nomi\land\neg\Box\Diamondblack\nomj))$ as $\phi$, then\\

$\forall p\forall\nomi(\nomi\leq\Diamond\Box p\ \&\ \phi\leq p\Rightarrow\nomi\leq\Box\Diamond\Box p)$\\
$\forall p\forall\nomi\forall\nomk(\nomi\leq\Diamond\nomk\ \&\ \nomk\leq\Box p\ \&\ \phi\leq p\Rightarrow\nomi\leq\Box\Diamond\Box p)$\\
$\forall p\forall\nomi\forall\nomk(\nomi\leq\Diamond\nomk\ \&\ \Diamondblack\nomk\leq p\ \&\ \phi\leq p\Rightarrow\nomi\leq\Box\Diamond\Box p)$\\
$\forall\nomi\forall\nomk(\nomi\leq\Diamond\nomk\Rightarrow\nomi\leq\Box\Diamond\Box (\Diamondblack\nomk\lor\phi))$\\

Then we can use standard translation to get its first-order correspondence.
\end{example}

\begin{example}\label{Example:irreflexivity}
$\forall q(\forall p(p\to\Diamond p\lor q)\to q)$\\
$\forall q\forall\nomi(\nomi\leq\forall p(p\to\Diamond p\lor q)\Rightarrow \nomi\leq q)$\\
$\forall q\forall\nomi(\forall p(\nomi\leq p\to\Diamond p\lor q)\Rightarrow \nomi\leq q)$\\
$\forall q\forall\nomi(\forall p(\nomi\leq p\Rightarrow \nomi\leq\Diamond p\lor q)\Rightarrow \nomi\leq q)$\\
$\forall q\forall\nomi(\nomi\leq\Diamond\nomi\lor q\Rightarrow \nomi\leq q)$\\
$\forall q\forall\nomi(\nomi\land\neg\Diamond\nomi\leq q\Rightarrow \nomi\leq q)$\\
$\forall\nomi(\nomi\leq\nomi\land\neg\Diamond\nomi)$\\
$\forall\nomi(\nomi\leq\neg\Diamond\nomi)$\\
$\forall x\neg Rxx$.

By \cite[Example 2.58]{BRV01}, the irreflexive property is not preserved under taking ultrafilter extensions, which means that the validity of $\forall q(\forall p(p\to\Diamond p\lor q)\to q)$ is not preserved under taking canonical extensions, which means that $\forall q(\forall p(p\to\Diamond p\lor q)\to q)$ is not canonical.
\end{example}

\begin{example}\label{Exa:Non-Sahlqvist}
The following example is not equivalent to any Sahlqvist formula in the basic modal language:\\
$\forall p(\Box p\land\forall q(q\to \Diamond\Diamond q\lor p)\to p)$\\
$\forall p\forall \nomi(\nomi\leq\Box p\land\forall q(q\to \Diamond\Diamond q\lor p)\Rightarrow \nomi\leq p)$\\
$\forall p\forall \nomi(\nomi\leq\Box p\ \&\ \nomi\leq\forall q(q\to \Diamond\Diamond q\lor p)\Rightarrow \nomi\leq p)$\\
$\forall p\forall \nomi(\Diamondblack\nomi\leq p\ \&\ \nomi\leq\forall q(q\to \Diamond\Diamond q\lor p)\Rightarrow \nomi\leq p)$\\
$\forall p\forall \nomi(\Diamondblack\nomi\leq p\ \&\ \forall q(\nomi\leq q\to \Diamond\Diamond q\lor p)\Rightarrow \nomi\leq p)$\\
$\forall p\forall \nomi(\Diamondblack\nomi\leq p\ \&\ \forall q(\nomi\leq q\Rightarrow\nomi\leq\Diamond\Diamond q\lor p)\Rightarrow \nomi\leq p)$\\
$\forall p\forall \nomi(\Diamondblack\nomi\leq p\ \&\ \nomi\leq\Diamond\Diamond\nomi\lor p\Rightarrow \nomi\leq p)$\\
$\forall p\forall \nomi(\Diamondblack\nomi\leq p\ \&\ \nomi\land\neg\Diamond\Diamond\nomi\leq p\Rightarrow \nomi\leq p)$\\
$\forall p\forall \nomi(\Diamondblack\nomi\lor(\nomi\land\neg\Diamond\Diamond\nomi)\leq p\Rightarrow \nomi\leq p)$\\
$\forall \nomi(\nomi\leq \Diamondblack\nomi\lor(\nomi\land\neg\Diamond\Diamond\nomi))$\\
$\forall \nomi(\nomi\leq \Diamondblack\nomi$ or $\nomi\leq\nomi\land\neg\Diamond\Diamond\nomi)$\\
$\forall \nomi(\nomi\leq \Diamondblack\nomi$ or $\nomi\leq\neg\Diamond\Diamond\nomi)$\\
$\forall \nomi(\nomi\leq\Diamond\Diamond\nomi\to\Diamondblack\nomi)$\\
$\forall x\forall y(Rxy\land Ryx\to Rxx)$\\

One can show that this property is not modally definable:\\

Consider $\mathbb{F}_{1}=(W_{1},R_{1})$ where $W_{1}$ is the set of all integers, $R_{1}=\{(x,x+1)\mid x\in W_{1}\}$, $\mathbb{F}_{2}=(W_{2}, R_{2})$ where $W_{2}=\{w_0,w_1\}$, $R_{2}=\{(w_0,w_1),(w_1,w_0)\}$, then $\mathbb{F}_{2}$ is a bounded morphic image of $\mathbb{F}_{1}$, $\mathbb{F}_{1}\vDash\forall x\forall y(Rxy\land Ryx\to Rxx)$, while $\mathbb{F}_{2}\nvDash\forall x\forall y(Rxy\land Ryx\to Rxx)$.
\end{example}

\subsection{$\Pi_2$-formulas and rules}

In this section we consider the following kinds of rules, each of which is the generalization of the former one:

\begin{itemize}
\item Gabbay's irreflexivity rule \cite{Ga81}:$$\vdash\neg(p\to\Diamond p)\to\phi\ \Rightarrow\ \vdash\phi$$
where $p$ does not occur in $\phi$.
\item Venema's non-$\xi$ rules \cite{Ve93}:$$\vdash\neg\xi(p_0,\ldots,p_n)\to\phi\ \Rightarrow\ \vdash\phi$$
where $p_0, \ldots, p_n$ does not occur in $\phi$.
\item $\Pi_2$ rules \cite{BeGhLa20}:$$\vdash F(\vec \phi/\vec x, \vec p)\to\chi\ \Rightarrow\ \vdash G(\vec \phi/\vec x)\to\chi$$
where $F,G$ are formulas, $\vec\phi$ is a tuple of formulas, $\chi$ is a formula, and $\vec p$ is a
tuple of propositional variables which do not occur in $\vec \phi$ and $\chi$.
\end{itemize}

\paragraph{Gabbay's irreflexivity rule.}Now consider Gabbay's irreflexivity rule, its corresponding $\forall\exists$-statement is the following:

$$\forall q(\forall p(\neg(p\to\Diamond p)\leq q)\ \Rightarrow\ \top\leq q)$$

therefore, its equivalent $\mathsf{SOPML}$ $\forall\exists$-formula is

$$\forall q(\forall p\ \mathbf{l}(\neg(p\to\Diamond p),q)\to\mathbf{l}(\top,q))$$

now its $\mathsf{ALBA}^{\mathsf{SOPML}}$-reduction is as follows:\footnote{Notice that the algorithm here is slightly different from the one defined in the previous sections, due to the introduction of the $\mathsf{l}$ connective in the basic language. Similar for the non-$\xi$ rules and the $\Pi_2$ rules.}\\

$\forall\nomi(\nomi\leq\forall q(\forall p\ \mathbf{l}(\neg(p\to\Diamond p),q)\to\mathbf{l}(\top,q)))$\\
$\forall\nomi\forall q((\forall p(\nomi\leq\mathbf{l}(\neg(p\to\Diamond p),q))\ \Rightarrow\ \nomi\leq\mathbf{l}(\top,q)))$\\
$\forall q(\forall p(\neg(p\to\Diamond p)\leq q)\ \Rightarrow\ \top\leq q)$\\
$\forall q(\forall p\forall\nomj(\nomj\leq\neg(p\to\Diamond p)\ \Rightarrow\ \nomj\leq q)\ \Rightarrow\ \top\leq q)$\\
$\forall q(\forall p\forall\nomj(\nomj\leq p\ \&\ \nomj\leq\neg\Diamond p\ \Rightarrow\ \nomj\leq q)\ \Rightarrow\ \top\leq q)$\\
$\forall q(\forall\nomj(\nomj\leq\neg\Diamond\nomj\ \Rightarrow\ \nomj\leq q)\ \Rightarrow\ \top\leq q)$\\
$\forall q(\forall\nomj(\mathbf{l}(\nomj,\neg\Diamond\nomj)\land\nomj\leq q)\ \Rightarrow\ \top\leq q)$\\
$\forall q(\exists\nomj(\mathbf{l}(\nomj,\neg\Diamond\nomj)\land\nomj)\leq q\ \Rightarrow\ \top\leq q)$\\
$\top\leq\exists\nomj(\mathbf{l}(\nomj,\neg\Diamond\nomj)\land\nomj)$\\
$\forall\nomi(\nomi\leq\exists\nomj(\mathbf{l}(\nomj,\neg\Diamond\nomj)\land\nomj))$\\
$\forall xST_{x}(\exists\nomj(\mathbf{l}(\nomj,\neg\Diamond\nomj)\land\nomj))$\\
$\forall x\exists jST_{x}((\mathbf{l}(\nomj,\neg\Diamond\nomj)\land\nomj))$\\
$\forall x\exists j(ST_{x}(\mathbf{l}(\nomj,\neg\Diamond\nomj))\land ST_{x}(\nomj))$\\
$\forall x\exists j(\neg Rjj\land x=j)$\\
$\forall x\neg Rxx$.

\paragraph{Venema's non-$\xi$ rules.}Now consider Venema's non-$\xi$ rules, their corresponding $\forall\exists$-statement is the following:

$$\forall q(\forall\vec p(\neg\xi(\vec p)\leq q)\ \Rightarrow\ \top\leq q)$$

When $\xi$ is a Sahlqvist formula $\mathsf{Sahl}_1(\vec p) \to \mathsf{POS}(\vec p)$ in the basic modal language, Venema's rules can be equivalently written in the following $\mathsf{SOPML}$ $\forall\exists$-formula:

$$\forall q(\forall\vec p(\mathbf{l}(\neg(\mathsf{Sahl}_1(\vec p) \to \mathsf{POS}(\vec p)), q))\to\mathbf{l}(\top, q)).$$

Assume that $\nomi\leq \mathsf{Sahl}_1(\vec p) \to \mathsf{POS}(\vec p)$ can be reduced to $\nomi\leq\mathsf{Local}$ where $\mathsf{Local}$ is pure (which is the modal counterpart of the local frame correspondent of $\mathsf{Sahl}_1(\vec p) \to \mathsf{POS}(\vec p)$), then the $\mathsf{ALBA}^{\mathsf{SOPML}}$-reduction is as follows:\\

$\forall q(\forall\vec p(\mathbf{l}(\neg(\mathsf{Sahl}_1(\vec p) \to \mathsf{POS}(\vec p)), q))\to\mathbf{l}(\top, q))$\\
$\forall\nomi(\nomi\leq\forall q(\forall\vec p(\mathbf{l}(\neg(\mathsf{Sahl}_1(\vec p) \to \mathsf{POS}(\vec p)), q))\to\mathbf{l}(\top, q)))$\\
$\forall\nomi(\nomi\leq\forall q(\forall\vec p(\mathbf{l}(\neg(\mathsf{Sahl}_1(\vec p) \to \mathsf{POS}(\vec p)), q))\to\mathbf{l}(\top, q)))$\\
$\forall\nomi\forall q(\nomi\leq\forall\vec p(\mathbf{l}(\neg(\mathsf{Sahl}_1(\vec p) \to \mathsf{POS}(\vec p)), q))\ \Rightarrow\ \nomi\leq\mathbf{l}(\top, q))$\\
$\forall\nomi\forall q(\forall\vec p(\nomi\leq\mathbf{l}(\neg(\mathsf{Sahl}_1(\vec p) \to \mathsf{POS}(\vec p)), q))\ \Rightarrow\ \nomi\leq\mathbf{l}(\top, q))$\\
$\forall\nomi\forall q(\forall\vec p(\neg(\mathsf{Sahl}_1(\vec p) \to \mathsf{POS}(\vec p))\leq q)\ \Rightarrow\ \nomi\leq\mathbf{l}(\top, q))$\\
$\forall q(\forall\vec p(\neg(\mathsf{Sahl}_1(\vec p) \to \mathsf{POS}(\vec p))\leq q)\ \Rightarrow\ \top\leq q)$\\
$\forall q(\forall\vec p\forall\nomj(\nomj\leq\neg(\mathsf{Sahl}_1(\vec p)\to\mathsf{POS}(\vec p))\ \Rightarrow\ \nomj\leq q)\ \Rightarrow\ \top\leq q)$\\
$\forall q(\forall\vec p\forall\nomj(\nomj\nleq(\mathsf{Sahl}_1(\vec p)\to\mathsf{POS}(\vec p))\ \Rightarrow\ \nomj\leq q)\ \Rightarrow\ \top\leq q)$\\
$\forall q(\forall\nomj(\nomj\nleq\mathsf{Local}\ \Rightarrow\ \nomj\leq q)\ \Rightarrow\ \top\leq q)$\\
$\forall q(\forall\nomj(\neg\mathbf{l}(\nomj,\mathsf{Local})\land\nomj\leq q)\ \Rightarrow\ \top\leq q)$\\
$\forall q(\exists\nomj(\neg\mathbf{l}(\nomj,\mathsf{Local})\land\nomj)\leq q\ \Rightarrow\ \top\leq q)$\\
$\top\leq\exists\nomj(\neg\mathbf{l}(\nomj,\mathsf{Local})\land\nomj)$\\
$\forall\nomi(\nomi\leq\exists\nomj(\neg\mathbf{l}(\nomj,\mathsf{Local})\land\nomj))$\\
$\forall xST_x(\exists\nomj(\neg\mathbf{l}(\nomj,\mathsf{Local})\land\nomj))$\\
$\forall x\exists jST_x(\neg\mathbf{l}(\nomj,\mathsf{Local})\land\nomj)$\\
$\forall x\exists j(ST_x(\neg\mathbf{l}(\nomj,\mathsf{Local}))\land ST_x(\nomj))$\\
$\forall x\exists j(ST_x(\neg\mathbf{l}(\nomj,\mathsf{Local}))\land x=j)$\\
$\forall x\exists j(\neg ST_j(\mathsf{Local})\land x=j)$\\
$\forall x\neg ST_x(\mathsf{Local})$.

\paragraph{$\Pi_2$-rules.}We first consider the corresponding $\mathsf{SOPML}$-formulas of $\Pi_2$-rules. For $\vdash F(\vec \phi/\vec x, \vec p)\to\chi\ \Rightarrow\ \vdash G(\vec \phi/\vec x)\to\chi$, its corresponding $\forall\exists$-statement is the following:

$$\forall \vec x\forall z(G(\vec x)\nleq z\ \Rightarrow\ \exists\vec y(F(\vec x,\vec y)\nleq z)),$$

which is equivalent to 

$$\forall \vec x\forall z(\forall\vec y(F(\vec x,\vec y)\leq z)\ \Rightarrow\ G(\vec x)\leq z),$$

which is essentially the following $\mathsf{SOPML}$ $\forall\exists$-formula:

$$\forall \vec p\forall q(\forall\vec r(\mathbf{l}(F(\vec p,\vec r), q))\to\mathbf{l}(G(\vec p), q))$$

When $F(\vec p,\vec r)$ is of the form $\mathsf{Sahl}_1(\vec p,\vec r)$, $G(\vec p)$ is of the form $\mathsf{POS}(\vec p)$, the $\mathsf{ALBA}^{\mathsf{SOPML}}$-reduction is as follows:\\

$\forall \vec p\forall q(\forall\vec r(\mathbf{l}(\mathsf{Sahl}_1(\vec p,\vec r), q))\to\mathbf{l}(\mathsf{POS}(\vec p), q))$\\
$\forall\nomi(\nomi\leq\forall \vec p\forall q(\forall\vec r(\mathbf{l}(\mathsf{Sahl}_1(\vec p,\vec r), q))\to\mathbf{l}(\mathsf{POS}(\vec p), q)))$\\
$\forall\nomi\forall \vec p\forall q(\nomi\leq\forall\vec r(\mathbf{l}(\mathsf{Sahl}_1(\vec p,\vec r), q))\ \Rightarrow\ \nomi\leq\mathbf{l}(\mathsf{POS}(\vec p), q))$\\
$\forall\nomi\forall \vec p\forall q(\forall\vec r(\nomi\leq\mathbf{l}(\mathsf{Sahl}_1(\vec p,\vec r), q))\ \Rightarrow\ \nomi\leq\mathbf{l}(\mathsf{POS}(\vec p), q))$\\
$\forall \vec p\forall q(\forall\vec r(\mathsf{Sahl}_1(\vec p,\vec r)\leq q)\ \Rightarrow\ \mathsf{POS}(\vec p)\leq q)$\\
$\forall \vec p\forall q(\forall\vec r\forall\nomi(\nomi\leq\mathsf{Sahl}_1(\vec p,\vec r)\ \Rightarrow\ \nomi\leq q)\ \Rightarrow\ \mathsf{POS}(\vec p)\leq q)$\\
$\forall \vec p\forall q(\forall\vec r\forall\nomi(\exists\vec \nomj(\mathsf{NEG}(\vec p,\vec r)\ \&\ \mathsf{NOM}\ \&\ \mathsf{MinVal}(\vec p,\vec r))\ \Rightarrow\ \nomi\leq q)\ \Rightarrow\ \mathsf{POS}(\vec p)\leq q)$\\

(Here $\mathsf{NEG}(\vec p,\vec r)\ \&\ \mathsf{NOM}\ \&\ \mathsf{MinVal}(\vec p,\vec r)$ are as described in Proposition \ref{Prop:Sahl_1}.)\\

$\forall \vec p\forall q(\forall\vec r\forall\nomi\forall\vec\nomj(\mathsf{NEG}(\vec p,\vec r)\ \&\ \mathsf{NOM}\ \&\ \mathsf{MinVal}(\vec p,\vec r)\ \Rightarrow\ \nomi\leq q)\ \Rightarrow\ \mathsf{POS}(\vec p)\leq q)$\\
$\forall \vec p\forall q(\forall\vec r\forall\nomi\forall\vec\nomj(\mathsf{NEG}(\vec p,\vec r)\ \&\ \mathsf{NOM}\ \&\ \mathsf{MinVal}(\vec p,\vec r)\ \Rightarrow\ \nomi\leq q)\ \Rightarrow\ \forall\nomk(\nomk\leq\mathsf{POS}(\vec p)\ \Rightarrow\ \nomk\leq q))$\\
$\forall \vec p\forall q\forall \nomk(\forall\vec r\forall\nomi\forall\vec\nomj(\mathsf{NEG}(\vec p,\vec r)\ \&\ \mathsf{NOM}\ \&\ \mathsf{MinVal}(\vec p,\vec r)\ \Rightarrow\ \nomi\leq q)\ \&\ \nomk\leq\mathsf{POS}(\vec p)\ \Rightarrow\ \nomk\leq q)$\\
$\forall \vec p\forall q\forall \nomk(\forall\vec r\forall\nomi\forall\vec\nomj(\nomk\leq\mathsf{POS}(\vec p)\ \&\ (\mathsf{NEG}(\vec p,\vec r)\ \&\ \mathsf{NOM}\ \&\ \mathsf{MinVal}(\vec p,\vec r)\ \Rightarrow\ \nomi\leq q))\ \Rightarrow\ \nomk\leq q)$\\

Then we can apply the Ackermann rule and substitute the minimal valuation of $\vec p,\vec r$ into $\mathsf{POS}(\vec p)$ and $\mathsf{NEG}(\vec p,\vec r)$ and make the latter two pure, therefore the complex inequality is equivalent to\\

$\forall q\forall \nomk(\forall\nomi\forall\vec\nomj(\mathsf{PURE}\ \&\ (\mathsf{PURE}'\ \Rightarrow\ \nomi\leq q))\ \Rightarrow\ \nomk\leq q)$\\

By the packing rule, $\mathsf{PURE}'\ \Rightarrow\ \nomi\leq q$ is packed into an inequality $\psi\leq q$ where $\psi$ is pure:\\

$\forall q\forall \nomk(\forall\nomi\forall\vec\nomj(\mathsf{PURE}\ \&\ \psi\leq q)\ \Rightarrow\ \nomk\leq q)$\\
$\forall q\forall \nomk(\forall\nomi\forall\vec\nomj(\mathsf{PURE})\ \&\ \forall\nomi\forall\vec\nomj(\psi\leq q)\ \Rightarrow\ \nomk\leq q)$\\
$\forall q\forall \nomk(\forall\nomi\forall\vec\nomj(\mathsf{PURE})\ \&\ \exists\nomi\exists\vec\nomj\psi\leq q\ \Rightarrow\ \nomk\leq q)$\\
$\forall \nomk(\forall\nomi\forall\vec\nomj(\mathsf{PURE})\ \Rightarrow\ \nomk\leq \exists\nomi\exists\vec\nomj\psi)$.\\

Then we can perform the standard translation to obtain its corresponding first-order correspondent.

\section{Conclusion}\label{Sec:Conclusion}

In this paper, we develop the Sahlqvist correspondence theory for $\mathsf{SOPML}$. We define the class of Sahlqvist formulas for $\mathsf{SOMPL}$, each formula of which is shown to have a first-order correspondent by an algorithm $\mathsf{ALBA}^{\mathsf{SOMPL}}$. In addition, we show that certain $\Pi_2$-rules correspond to $\Pi_2$-Sahlqvist formulas in $\mathsf{SOMPL}$, which further correspond to first-order conditions.

Here we give some final remarks:

\begin{itemize}
\item Since the Sahlqvist correspondence theorem talks about frame definability, any propositional variables in the basic modal formulas are already implicitly treated as universally quantified, so what we will do in this paper for $\mathsf{SOPML}$ formulas is to find a Sahlqvist fragment which allows also for existentially quantified proposition variables, not only universally quantified variables. Indeed, this can be seen in the definition of $\Pi_2$-Sahlqvist formulas, where universal quantifiers are allowed in the antecedent part.

\item This paper can also be seen as looking for a modal counterpart of second-order quantifier elimination for monadic second-order logic ($\mathsf{MSO}$), as $\mathsf{SOPML}$ with global modality is expressively equivalent to $\mathsf{MSO}$ (see \cite{Ku08}). Here what we are aiming at is to find a natural fragment in a modal-type language which can be reduced to first-order formulas.
\end{itemize}

\paragraph{Acknowledgement} The research of the author is supported by Taishan University Starting Grant ``Studies on Algebraic Sahlqvist Theory'' and the Taishan Young Scholars Program of the Government of Shandong Province, China (No.tsqn201909151). The author would like to thank Nick Bezhanishvili for his suggestions and comments on this project, and Balder ten Cate for the detailed comments and remarks which help in improving the paper.

\bibliographystyle{abbrv}
\bibliography{Pi_2}
\end{document}